\documentclass[11pt]{article}
\usepackage{times}
\usepackage{anysize}
\marginsize{1in}{1in}{1in}{1in}

\usepackage{microtype}
\usepackage{graphicx}
\usepackage{booktabs} 
\usepackage{caption}
\usepackage{subcaption}
\usepackage{lineno}
\usepackage{hyperref}

\usepackage{amsmath}
\usepackage{amssymb}
\usepackage{mathtools}
\usepackage{amsthm}
\usepackage[capitalize,noabbrev]{cleveref}
\usepackage[utf8]{inputenc} 
\usepackage{comment}
\usepackage{cleveref}
\usepackage[justification=centering]{caption}
\usepackage{longtable}
\usepackage{tabularx,booktabs}
\usepackage{multirow}
\usepackage{algorithm}
\usepackage[noend]{algorithmic}
\allowdisplaybreaks[4]
\theoremstyle{plain}
\newtheorem{theorem}{Theorem}

\newtheorem{lemma}[theorem]{Lemma}

\theoremstyle{definition}
\newtheorem{definition}[theorem]{Definition}

\theoremstyle{remark}

\begin{document}

\title{Data-dependent Evaluations for Budgeted Submodular Maximization} 

  \author{
    Lejian Zhang\textsuperscript{\rm 1},
    Xueyan Tang\textsuperscript{\rm 1},
    and Jing Tang\textsuperscript{\rm 2}\\
    \textsuperscript{\rm 1}
    College of Computing and Data Science,
    Nanyang Technological University, Singapore\\
    \textsuperscript{\rm 2}
    Data Science and Analytics Thrust, 
    The Hong Kong University of Science \\ and Technology (Guangzhou), China\\
    lejian001@e.ntu.edu.sg, asxytang@ntu.edu.sg, jingtang@ust.hk
  }
  \date{}
  \maketitle

\begin{abstract}
Submodular maximization is an important building block for developing algorithms in many areas such as machine learning and data mining. Due to the NP-hardness of the problem, analysis of submodular maximization algorithms typically provides pessimistic worst-case approximation factors only. It is not easy to evaluate how close a produced solution is to an optimal one for a given problem instance. In this paper, we develop new data-dependent upper bounds for submodular maximization with a knapsack constraint. We theoretically prove that they dominate the optimal solution and empirically demonstrate their advantages in certifying how close to optimal a solution is through experiments with real-world datasets.
\end{abstract}

\section{Introduction}

A great number of optimization problems in different areas can be modeled as submodular maximization problems, including facility location~\cite{lindgren2016leveraging}, feature selection~\cite{das2011submodular}, recommendation system~\cite{leskovec2007cost}, influence maximization~\cite{tang2018online}, document summarization~\cite{lin2010multi}, maximum coverage~\cite{feige1998threshold}, sensor placement~\cite{krause2008near}, exemplar sampling~\cite{gomes2010budgeted} and market expansion~\cite{dughmi2009revenue}.

    Due to the NP-hard nature of submodular maximization, 
    many approximation algorithms have been developed, attempting to find good solutions in polynomial time. 
    Approximation algorithms are typically evaluated by the {\em approximation factor}, which refers to the worst-case ratio between the function value of an output solution and that of an optimal solution. 
    Nemhauser and Wolsey
    \cite{nemhauser1978best} introduced a simple greedy algorithm which guarantees an approximation factor of $1-1/e$ for the problem of Monotone Submodular Maximization with a Cardinality constraint (MSMC) (finding a set $S$ of a given size which maximizes $f(S)$). But the greedy algorithm does not guarantee any positive approximation factor for the more general problem of Monotone Submodular Maximization with a Knapsack constraint (MSMK) (finding a set $S$ of total cost at most a given budget which maximizes $f(S)$, assuming elements have costs). 
    Wolsey
    \cite{wolsey1982maximising} improved the simple greedy algorithm by a small modification and showed that it achieves an approximation factor of 0.357. 
    Khuller et al. \cite{KHULLER199939} attempted to derive an approximation factor of $1-1/\sqrt{e} \approx 0.393$ for the modified greedy algorithm, but their analysis was found flawed~\cite{10.1145/3219819.3219946}.
    Tang et al.
    \cite{tang2021revisiting} 
    proved that this modified greedy algorithm actually guarantees an approximation factor of  0.405. Subsequently, more careful approximation analysis showed that the exact approximation factor of the modified greedy algorithm is between 0.427 and 0.42945~\cite{Feldman2023,Kulik2021}. Both the simple and modified greedy algorithms run in $O(n^2)$ time, where $n$ is the size of the ground set. 
    Sviridenko
    \cite{sviridenko2004note} proposed a $(1-1/e)$-approximation algorithm at the expense of increasing the time complexity to $O(n^5)$. 
    The time complexity of the algorithm was later reduced to $O(n^4)$ without sacrificing the approximation factor~\cite{Feldman2023,Kulik2021}.
    Yaroslavtsev et al.
    \cite{yaroslavtsev2020bring} developed a $0.5$-approximation algorithm called Greedy+Max that runs in $O(n^2)$ time. 
    Feldman et al.
    \cite{Feldman2023} augmented Greedy+Max with a single guess to achieve an approximation factor of $0.6174$ in $O(n^3)$ running time. 
    The running times of these algorithms can be reduced from $O(n^i)$ to $\tilde{O}(n^{i-1}/ \epsilon)$ (ignoring poly-logarithmic terms) at the cost of losing an $\epsilon$ in the approximation factors using the threshold technique of \cite{badanidiyuru2014fast}.

    Although the above algorithms provide approximation factors, these constant factors usually leave a significant gap between the solution obtained and the optimal solution of a particular problem instance because they consider only the worst-case scenarios. It is challenging to measure how close to optimal an output solution is for a specific problem instance in practice. One possible method is to derive upper bounds on the optimal function values on a per-instance basis. We refer to such bounds as {\em data-dependent bounds}. Unlike the constant approximation factors which are fixed and independent of problem instances, data-dependent bounds can be used to characterize the actual quality of the solutions constructed on different problem instances in an instance-aware manner. 
    Recently, a few studies have attempted this method.
    Balkanski et al.
    \cite{balkanski2021instance} and Chakrabarty and Cote \cite{chakrabarty2023primal} introduced instance-specific upper bounds 
    for the MSMC problem, but these bounds cannot be applied to the MSMK problem directly. 
    Tang et al.
    \cite{tang2021revisiting} presented a naive data-dependent upper bound 
    for the MSMK problem. 

    In this paper, we construct and evaluate new data-dependent upper bounds on the optimal solution for the MSMK problem. Our contributions can be summarized as follows: 
    \begin{itemize}
        \item We develop two strategies, a slicing strategy and a removing strategy, to derive data-dependent upper bounds. We theoretically prove that the constructed bounds dominate the optimal solution, and they are tighter than the existing bounds.
        \item We transform these strategies into linear programs so that multiple base sets can be incorporated into the derivation of data-dependent upper bounds, which further tightens the established bounds. 
        \item We conduct extensive experiments with several real-world applications, including maximum coverage, revenue maximization, and feature selection, to demonstrate empirically the advantages of our proposed bounds in certifying the quality of solutions to the MSMK problem. 
    \end{itemize}

The rest of this paper is organized as follows. \Cref{sec:pre} gives a formal definition of the MSMK problem and some preliminaries. \Cref{sec:des} elaborates the design and analysis of our data-dependent upper bounds.
\Cref{sec:exp} discusses the experiments.
Finally, \Cref{sec:con}
concludes the paper.
\section{Preliminaries}
\label{sec:pre}
    
    Given a ground set $V$, a set function $f\colon 2^{V} \rightarrow \mathbb{R}$ is {\em submodular} if for any two sets $A, B \subseteq V$, it holds that
\begin{equation*}
    f(A) + f(B) \ge f(A \cap B) + f(A \cup B).
\end{equation*}
    An equivalent definition of a submodular set function $f$ is that for any two sets $A \subseteq B \subseteq V$ and any element $v \in V \setminus B$, it holds that 
\begin{equation*}
    f(A \cup \{v\}) - f(A) \ge f(B \cup \{v\}) - f(B).
\end{equation*}
    The latter definition describes the {\em diminishing return} property of a submodular set function.

    In this paper, we focus on monotone submodular maximization, where a set function $f$ is {\em monotone} (non-decreasing) if $f(A) \le f(B)$ for any two sets $A \subseteq B \subseteq V$. 
    The input to the MSMK problem includes a ground set $V$, a non-negative monotone submodular set function $f:2^{V} \rightarrow \mathbb{R}_{\ge 0}$, a non-negative modular cost function $c:2^V \rightarrow \mathbb{R}_{\ge 0}$ to measure the cost of selecting elements (where $c(S) = \sum_{v \in S} c(v)$ and $c(v)$ is the cost of element $v$), and a budget $b$. The objective of the MSMK problem is to find a set that maximizes the function value $f$ among all the sets of cost at most $b$, i.e.,
    $\mathop{\arg\max}_{S \subseteq V,\, c(S) \le b} f(S)$.
    Without loss of generality, we assume that for each element $v \in V$, $c(v) \le b$.

    We define the {\em marginal gain} of an element $v$ with respect to a set $S$ as
    $f_S(v) = f(S \cup \{v\}) - f(S)$,
    i.e., the increase in the function value by adding $v$ to $S$.
    We define the {\em marginal density} of an element $v$ with respect to a set $S$ as
    $d_S(v) = \frac{f_S(v)}{c(v)}$,
    i.e., the marginal gain normalized by the element cost. Due to the diminishing return property of a submodular set function, the marginal gain and density of an element $v$ are the highest when $S = \emptyset$ and are the lowest when $S = V \setminus \{v\}$ (among the sets $S$ that do not contain $v$). 
    We refer to the marginal density in the latter case, i.e., $d_{V \setminus \{v\}}(v)$, as the {\em cutoff density}. For notational convenience, we shall use $f_S(T)$ to denote the marginal gain of a set $T$ with respect to another set $S$, 
    i.e.,
    $f_S(T) = f(S \cup T) - f(S)$.

    Let $OPT$ denote an optimal solution to the MSMK problem. For any set $S \subseteq V$, we have
\begin{align}
\label{equ:upb0}
f(OPT) & = f(OPT \cup S) - f_{OPT}(S) \nonumber \\ & = f(S) + f_S(OPT) - f_{OPT}(S) \nonumber \\ & \le f(S) + \max_{T \subseteq V,\, c(T) \le b} (f_S(T) - f_T(S)) \nonumber \\ & \le f(S) + \max_{T \subseteq V,\, c(T) \le b} f_S(T), 
\end{align}
where the first inequality follows from the fact that $c(OPT) \le b$.
    Our main approach is to construct an upper bound $\Lambda (S)$ on the last term $\max_{T \subseteq V,\, c(T) \le b} (f_S(T) - f_T(S))$ or $\max_{T \subseteq V,\, c(T) \le b} f_S(T)$ in \Cref{equ:upb0}. 
    Given any set $S \subseteq V$, we can then derive an upper bound on $f(OPT)$ by computing $f(S) + \Lambda(S)$,     
    where $S$ will be termed the {\em base set}.
    Technically, we develop two strategies: a slicing strategy that constructs an upper bound on $\max_{T \subseteq V,\, c(T) \le b} f_S(T)$, and a removing strategy that derives an upper bound on $\max_{T \subseteq V,\, c(T) \le b} (f_S(T) - f_T(S))$ from an upper bound on $\max_{T \subseteq V,\, c(T) \le b} f_S(T)$. 
    We further transform these strategies into linear programs to facilitate  
    dealing with multiple base sets. 

    For simplicity of presentation, we define the ground set $V := \{1,2,\dots,n\}$. 
    Without loss of generality, we assume that $S=\{1,2,\dots,p\}$ when discussing a base set $S \subseteq V$.
    In addition, we also assume that the elements in $S$ are arranged in non-descending order of cutoff density, i.e., $d_{V\setminus\{1\}}(1) \le d_{V\setminus\{2\}}(2) \le \cdots \le d_{V\setminus\{p\}}(p)$. The elements in the remaining set $V \setminus S$ are arranged in non-ascending order of marginal density with respect to $S$, i.e., $d_S(p+1) \ge d_S(p+2) \ge \cdots \ge d_S(n)$.

\section{Design and analysis of upper bounds}
\label{sec:des}

\subsection{Removing strategy}

    It is easy to infer that $\max_{T \subseteq V,\, c(T) \le b} f_S(T) = \max_{T \subseteq V \setminus S,\, c(T) \le b} f_S(T)$, because $T$ should include only elements from $V \setminus S$ in order to maximize the marginal gain $f_S(T)$. To construct an upper bound on $\max_{T \subseteq V \setminus S,\, c(T) \le b} f_S(T)$,
    an intuitive idea is to find the lowest index $r$ in the remaining set $V \setminus S = \{p+1, p+2,\dots,n\}$ satisfying $\sum_{i=p+1}^{r}c(i) > b$ \cite{tang2021revisiting}. The upper bound is 
    defined as follows:
\begin{equation}
\label{equ:upb1}
    \Lambda^0(S) = \sum_{i=p+1}^{r-1}f_S(i) + \left(b-\sum_{i=p+1}^{r-1}c(i)\right)\cdot d_S(r).  
\end{equation}
    In essence, $\Lambda^0$ takes elements in the remaining set greedily according to their marginal density with respect to the base set $S$ until the budget $b$ is fulfilled. 
    Owing to the submodularity of $f$, it is easy to see that 
    $\Lambda^0(S)$ is an upper bound on $\max_{T \subseteq V \setminus S,\, c(T) \le b} f_S(T) = \max_{T \subseteq V,\, c(T) \le b} f_S(T)$.
    Given a base set $S$, the time complexity to compute $\Lambda^0(S)$ is $O(n \log n) = O(|V| \log |V|)$ due to the sorting of the elements.

    By \Cref{equ:upb0}, $f(S) + \Lambda^0(S)$ gives an upper bound on $f(OPT)$. Nevertheless, it ``consumes'' a total cost of $b+c(S)$. On the other hand, it is impossible for an optimal solution to spend any cost exceeding the budget $b$. 
    We would like to tighten $\Lambda^0$ by removing some elements and make the total cost equal to $b$, i.e., to construct an upper bound on $\max_{T \subseteq V, c(T) \le b} (f_S(T) - f_T(S))$.
    
    We start by defining a continuous version of the marginal gain function with respect to $S$ based on the remaining set $V \setminus S = \{p+1, p+2,\dots,n\}$:  
    \begin{equation}
    \label{eq:G+}
        G_+(x) = \sum_{i=p+1}^{r-1}f_S(i) + \left(x-\sum_{i=p+1}^{r-1}c(i)\right)\cdot d_S(r),
    \end{equation}
    where $r$ is the lowest index satisfying $\sum_{i=p+1}^{r}c(i) > x$.
    
    Note that $\Lambda^0(S)=G_+(b)$. $G_+(x)$ calculates an upper bound on the marginal gain generated by selecting elements from $V \setminus S$ with a total cost of $x$, i.e., for any $T\subseteq V\setminus S$,
    \begin{equation}
    \label{eq:G+property}
        G_+(c(T))\geq \sum_{i\in T}f_S(i) \geq f_S(T).
    \end{equation}
    Obviously, $G_+(x)$ is non-decreasing. Intuitively, we would like to pick $G_+(b-c(S))$ to fill up the budget $b$. However, $G_+(b-c(S))$ is not necessarily an upper bound on $f_S(OPT\setminus S)$, because we cannot guarantee that $c(OPT\setminus S) \le b-c(S)$.
    
    We propose to first choose elements to consume a total cost of $b+x$ and then remove elements from $S$ to reduce the total cost to $b$. To do so, we define another continuous function $G_-(x)$ to calculate a lower bound on the loss in the function value $f$ caused by removing elements with a total cost of $x$.
    Function $G_-(x)$ is defined as:
\begin{align}
    G_-(x) &= \sum_{i=1}^{s-1}f_{V\setminus\{i\}}(i) 
    + \left(x-\sum_{i=1}^{s-1} c(i)\right)\cdot d_{V\setminus\{s\}}(s), 
    \label{eq:gminus} 
\end{align}
    where $s$ is the lowest index in $S=\{1,2,\dots,p\}$ satisfying $\sum_{i=1}^{s}c(i) > x$. That is, $G_-(x)$ takes elements greedily according to their cutoff density until an accumulated cost $x$ is reached. Clearly, $G_-(x)$ is also non-decreasing. It is easy to see that for any set $T \subseteq S\setminus OPT$, we have 
    \begin{equation*}
        G_-(c(T)) \leq \sum_{i \in T} f_{V \setminus \{i\}}(i) \leq f_{V\setminus T}(T) \leq f_{OPT}(T),
    \end{equation*}
    where the first inequality is by the definition of $G_-(x)$ and the last two inequalities are due to the submodularity of $f$.

    To maintain a total cost equal to $b$, if $G_+(b-c(S)+x)$ is picked, we can remove $G_-(x)$ from it.
    Ideally, we would like to set 
    \begin{equation*}
        x^\ast = \max\{c(OPT\cup S) - b, 0\}.
    \end{equation*}
    Then, it holds that $b - c(S) + x^\ast \ge b - c(S) + c(OPT\cup S) - b = c(OPT\cup S) - c(S)$ and $x^\ast \le c(OPT\cup S) - c(OPT)$. 
    Hence, we have
\begin{align*}
    &
    G_+(b - c(S) + x^\ast) - G_-(x^\ast) \\ 
    & \geq  G_+(c(OPT\cup S) - c(S)) 
    - G_-(c(OPT\cup S) - c(OPT)) \\
    & =  G_+(c(OPT\setminus S)) - G_-(c(S\setminus OPT)) \\ 
    & \geq  f_S(OPT\setminus S) - f_{OPT}(S\setminus OPT) \\ 
    & =  f(OPT) - f(S).
\end{align*}
    However, $c(OPT\cup S)$ is not known as it is dependent on the optimal solution $OPT$. Note that $x^\ast \le \max\{c(OPT) + c(S) - b, 0\} \le c(S)$ always holds. Thus, we can pick the maximal value of $G_+(b-c(S)+x) - G_-(x)$ among all $x \in [0, c(S)]$.
    Hence, an upper bound incorporating the removing strategy is defined as
\begin{equation*}
    \Lambda^{1}(S) 
    = \max\limits_{x \in [0, c(S)]}\left\{G_+(b-c(S)+x) - G_-(x)\right\}. 
\end{equation*}
    
Since $G_+(x)$ is a non-decreasing function, we have
\begin{align*}
    \Lambda^0(S) &= G_+(b) 
    \ge \max\limits_{x \in [0, c(S)]}G_+(b-c(S)+x) 
    \\ &
    \ge \max\limits_{x \in [0, c(S)]}(G_+(b-c(S)+x)-G_-(x)) 
    = \Lambda^1(S),
\end{align*}
    which shows that $\Lambda^1$ is at least as good as $\Lambda^0$.

\subsubsection{$\Lambda^1$ in linear program presentation}
\label{sec:lbd 1}
    We can also present the calculation of $\Lambda^{1}$ as a linear program shown in \Cref{lnp:lbd 1}, where $\langle x_1, x_2, \dots, x_n \rangle \in [0,1]^n$ is an indicator vector.\footnote{All linear programs in this paper will have variables $x_1, x_2, \dots, x_n$ forming an indicator vector. We shall omit the constraint $\forall i \in V,\, 0 \le x_i \le 1$ in the presentations of linear programs for conciseness.} 
\begin{align}
\label{lnp:lbd 1}
    \max & \quad \sum_{i \in V\setminus S} f_S(i) \cdot x_i - \sum_{i \in S} f_{V \setminus \{i\}}(i) \cdot (1-x_i)\\
    \text{subject to} & \quad \sum_{i \in V} c(i) \cdot x_i \le b. \nonumber 
\end{align}
    We prove the equivalence between the optimum of \Cref{lnp:lbd 1} and $\Lambda^1(S)$ in Appendix \ref{appendix_lp_lambda1}.

    According to \Cref{lnp:lbd 1}, we can arrange elements in $V$ in descending order according to their ``weights'' in the objective function. That is, for each $i \in S$, its weight is $f_{V \setminus \{i\}}(i)$; for each $i \in V\setminus S$, its weight is $f_S(i)$. We can then pick elements greedily from the beginning and stop when the budget is fulfilled. Given a set $S$, the time complexity to compute $\Lambda^1(S)$ is $O(n \log n) = O(|V| \log |V|)$ due to the sorting of the elements. Thus, $\Lambda^1(S)$ has the same time complexity to compute as $\Lambda^0(S)$.

Similarly, we can also write the calculation of $\Lambda^0$ as a linear program shown in \Cref{lnp:lbd 0}.
\begin{align}
\label{lnp:lbd 0}
    \max & \quad \sum_{i \in V\setminus S} f_S(i) \cdot x_i \\
    \text{subject to} & \quad \sum_{i \in V} c(i) \cdot x_i \le b. \nonumber 
\end{align}

Notice that \Cref{lnp:lbd 1} differs from \Cref{lnp:lbd 0} in only one item: $-\sum_{i \in S}f_{V\setminus \{i\}}(i) \cdot (1-x_i)$. We refer to this item as the ``removing item'', which is in essence the core of the removing strategy.
\subsection{Slicing strategy}

We develop a slicing strategy to derive a data-dependent upper bound for the MSMK problem, which is inspired by Balkanski et al.'s work \cite{balkanski2021instance}.
Their method upper bounds the optimal MSMC solution 
by lower bounding the dual objective of 
finding the set $S$ of minimum size that has a given function value.
We introduce it
from a different perspective to motivate our approach to deal with the MSMK problem.

Recall that the remaining set $V \setminus S = \{p+1,p+2,\dots,n\}$ is sorted in non-ascending order of marginal density (or equivalently marginal gain in the case of MSMC) with respect to $S$. 
To simplify presentation, we define $V_i := \{1,2,\dots,i\}$ and $V_0 := \emptyset$.

Let $k$ be the cardinality constraint. 
If the upper bound $\Lambda^0$ is applied, we pick the top $k$ elements $p+1,p+2,\dots,p+k$ from $V \setminus S$ which have the largest marginal densities with respect to $S$. 
Now we arrange the elements in $OPT \setminus S$ in non-ascending order of marginal density with respect to $S$ and obtain a sorted sequence $\langle o_1, o_2, \dots, o_k \rangle$ where $o_1 < o_2 < \cdots < o_k$. Similarly, we define $O_i := \{o_1, o_2, \dots, o_i\}$ and $O_0 := \emptyset$.
Let us first notice two properties of $\langle o_1, o_2, \dots, o_k \rangle$ as follows:
\begin{itemize}
    \item for each $i \in \{1,2,\dots,k\}$, $f_{S \cup O_{i-1}}(o_i) \le f_S(o_i)$, which is due to the submodularity of $f$;
    \item for each $i \in \{1,2,\dots,k\}$, $f_{S \cup O_{i-1}}(o_i) = f_S(O_i) - f_S(O_{i-1}) \le f_S(\{p+1,p+2,\dots,o_i\}) - \sum_{j=1}^{i-1} f_{S \cup O_{j-1}}(o_j) = f_S(V_{o_i}) - \sum_{j=1}^{i-1} f_{S \cup O_{j-1}}(o_j)$, which is due to the monotonicity of $f$.
\end{itemize}
Balkanski et al.'s method \cite{balkanski2021instance} replaces each $f_{S \cup O_{i-1}}(o_i)$ with a variable $v_i$. Each variable $v_i$ corresponds to a distinct element in $V \setminus S$ because $o_i$ can be any element in $V \setminus S$. 
It then maximizes the total value of the sequence $v_1, v_2, \dots, v_k$ satisfying the above conditions. Since $f_S(o_1), f_{S \cup O_{1}}(o_2), \dots, f_{S \cup O_{k-1}}(o_k)$ is also such a sequence satisfying these conditions, the maximum total value of $v_1, v_2, \dots, v_k$ is an upper bound for it. Thus, $\sum_{i=1}^{k}v_i \ge \sum_{i=1}^{k}f_{S \cup O_{i-1}}(o_i) = f_S(OPT \setminus S)$, and $f(S) + \sum_{i=1}^{k}v_i$ is an upper bound on $f(OPT)$.
The elements picked by this
method do not necessarily have successive indexes. Hence,
it produces an upper bound at least as good as (no larger than) $\Lambda^0$. 

Remember that we would like to construct an upper bound on $\max_{T \subseteq V,\, c(T) \le b} f_S(T) = \max_{T \subseteq V \setminus S,\, c(T) \le b} f_S(T)$ 
for the MSMK problem. 
The main challenge of MSMK is that the elements have non-uniform costs so that the number of elements to pick and fill up the budget $b$ is not predetermined. A naive approach to address it is to choose a fixed value $\epsilon$ and divide the budget $b$ into slices of cost $\epsilon$ each, as illustrated in \Cref{fig:ub2s}. Then, to fill up the budget $b$, we pick $\frac{b}{\epsilon}$ slices. However, the element costs are not necessarily multiples of $\epsilon$. Hence, it is possible for a slice to involve a number of elements, where the number is not fixed and the first as well as last elements may be partially involved. Note that all the elements involved in a slice would have successive indexes. In general, it would not guarantee that the maximum total value of all slices is an upper bound on $f_S(OPT \setminus S)$, 
because the elements in $OPT \setminus S$ may not have successive indexes.

To address this issue, we can choose $\epsilon \rightarrow 0$ so that each slice would involve only one element
and thus share the same marginal density with the element it belongs to. 
The drawback is that setting $\epsilon \rightarrow 0$ would increase the number of slices to pick towards infinity and hence may significantly increase the time complexity of computing the upper bound. To guarantee efficiency, we can pick all the slices belonging to the same element at once.

Motivated by these thoughts, we present a novel and elegant design to 
construct an upper bound $\Lambda^2$ for the MSMK problem. 

We start by defining the concept of a valid partition for the MSMK problem.  

\begin{definition}
\label{def:partitionnew}
    Given a budget $b$ and a set $S$, a sequence of values $v_1,v_2,\dots,v_n \in \mathbb{R}_{\ge 0}$ form a valid partition of a value $v \in \mathbb{R}_{\ge 0}$ with respect to $S$ if
\begin{itemize}    
    \item[(i)] $\sum_{i=1}^n v_i = v$;
\end{itemize}
    and there exists a sequence of costs $s_1, s_2, \dots, s_n \in \mathbb{R}_{\ge 0}$ where $s_i \le c(i)$ represents the cost spent on element $i$ such that
\begin{itemize}
    \item[(ii)] $\sum_{i=1}^n s_i \le b$;
    \item[(iii)] for each $i \in \{1,2,\dots,n\}$, $d_S(i) \cdot s_i = \frac{f_S(i)}{c(i)} \cdot s_i \ge v_i$;
    \item[(iv)] for each $i \in \{1,2,\dots,n\}$, $f_S(V_i) \ge \sum_{j=1}^{i} v_j$.
\end{itemize}
\end{definition}

\begin{figure}[t]
    \centering
    \begin{minipage}[t]{0.45\textwidth}
        \centering
        \includegraphics[width=1\textwidth]{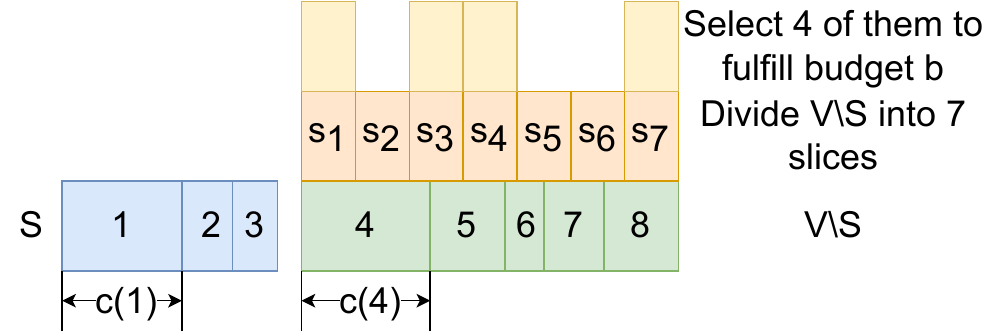}
        \caption{Naive slicing strategy}
        \label{fig:ub2s}
    \end{minipage}
    \hfill
    \begin{minipage}[t]{0.5\textwidth}
        \centering
        \includegraphics[width=1\textwidth]{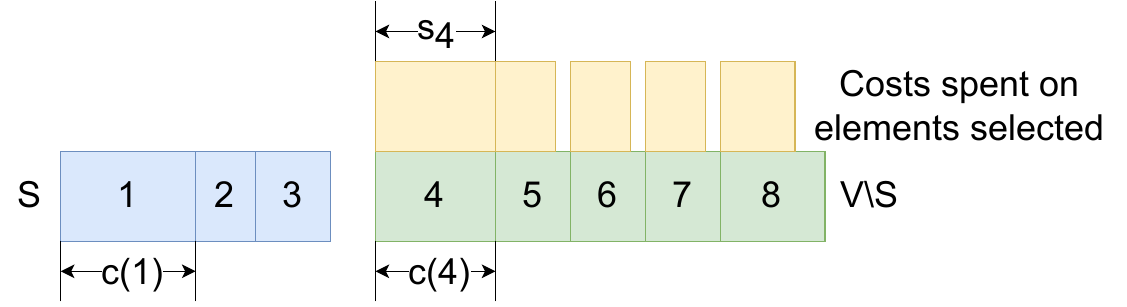}
        \caption{Slicing strategy producing $\Lambda^2$}
        \label{fig:ub2}
    \end{minipage}
\end{figure}

In Definition \ref{def:partitionnew}, $s_i$ can be interpreted as the total cost of the slices we pick involving element $i$, as illustrated in \Cref{fig:ub2}. 
Condition (ii) states that the total cost of all the slices picked is capped by $b$. 
Note that for each element $i \in S$, $d_S(i) = 0$. By condition (iii), we always have $v_i=0$ for each $i \in \{1,2,\dots,p\}$.

First, we prove the following property. 
\begin{lemma}
\label{lem:uboffT}
For any set $T \subseteq V \setminus S$ with $c(T) \le b$, $f_S(T)$ has a valid partition with respect to $S$.
\end{lemma} 
\begin{proof}[Proof Sketch]
We construct a valid partition of $f_S(T)$ to prove this lemma. 
Please refer to Appendix \ref{appendix_proof_of_lem:uboffT} for details.
\end{proof}

Lemma \ref{lem:uboffT} implies that the maximum value that has a valid partition with respect to $S$ is an upper bound on $\max_{T \subseteq V \setminus S,\, c(T) \le b} f_S(T)$.
We remark that Definition \ref{def:partitionnew} and Lemma \ref{lem:uboffT} actually do not require the remaining set $V \setminus S = \{p+1,p+2,\dots,n\}$ to be sorted in non-ascending order of marginal density with respect to $S$. 
On the other hand, having $V \setminus S$ sorted allows us to develop an efficient algorithm to compute the maximum value that has a valid partition.

To find this maximum value, 
for each element $i \in V \setminus S$, we would like to spend as little cost $s_i$ as possible while maximizing $v_i$.
By conditions (iii) and (iv) of Definition \ref{def:partitionnew}, in general, the cost to spend on each element $a_i$ should be $s_i = \frac{d_{S \cup V_{i-1}}(i)}{d_S(i)}\cdot c(i)$, so that the increase from $f_S(V_{i-1})$ to $f_S(V_{i})$ is the same as $d_S(i) \cdot s_i$.
Thus, we iterate through the sorted sequence $\langle p+1,p+2,...,n\rangle$ until encountering an element $t$ where the remaining budget $b - \sum_{i=p+1}^{t-1} \frac{d_{S \cup V_{i-1}}(i)}{d_S(i)}\cdot c(i) \le \frac{d_{S \cup V_{t-1}}(t)}{d_S(t)}\cdot c(t)$. 
For each element $i$ $(p+1 \le i < t)$, a cost of $s_i = \frac{d_{S \cup V_{i-1}}(i)}{d_S(i)}\cdot c(i)$ is spent on $i$ and we set $v_i = d_S(i) \cdot s_i = f_{S \cup V_{i-1}}(i)$.
For the element $t$, a cost of $s_t = b-\sum_{i=p+1}^{t-1}\frac{d_{S \cup V_{i-1}}(i)}{d_S(i)}\cdot c(i)$ is spent on $t$ and we set $v_t = d_S(t) \cdot s_t$.
For all the remaining elements $i$ ($1 \le i \le p$ or $t+1 \le i \le n$), we set $s_i = 0$ and $v_i = 0$.
Then, we define 
\begin{align}
    & 
    \Lambda^2(S) = \sum_{i=1}^{n} v_i = \sum_{i=p+1}^{t} v_i 
    = f_S(V_{t-1}) 
    + \left(b-\sum_{i=p+1}^{t-1}\frac{d_{S \cup V_{i-1}}(i)}{d_S(i)}\cdot c(i)\right)\cdot d_S(t). \label{eq:Lambda2}
\end{align}
The idea of constructing $\Lambda^2$ is illustrated in \Cref{fig:ub2}. 
Given a set $S$, the time complexity to compute $\Lambda^2(S)$ is $O(n \log n) = O(|V| \log |V|)$ due to the sorting of the elements.
Obviously, $\Lambda^2(S)$ has a valid partition as demonstrated by the cost and value sequences above.
Next, we prove that $\Lambda^2(S)$ is the maximum $v$ that has a valid partition and hence an upper bound on $\max_{T \subseteq V \setminus S,\, c(T) \le b} f_S(T) = \max_{T \subseteq V,\, c(T) \le b} f_S(T)$.

Notice that
    $\Lambda^2(S) 
    =\sum_{i=p+1}^{t} v_i
    = \sum_{i=p+1}^{t} d_S(i) \cdot s_i$.
Since $d_S(p+1) \ge d_S(p+2) \ge \cdots \ge d_S(n)$ and $\sum_{i=p+1}^{t}s_i = b$, we have $\sum_{i=p+1}^{t}d_S(i) \cdot s_i \le \sum_{i=p+1}^{r-1}d_S(i) \cdot c(i) + \left( b-\sum_{i=p+1}^{r-1}c(i) \right) \cdot d_S(r) = \Lambda^0(S)$, where $r$ is the lowest index satisfying $\sum_{i=p+1}^{r}c(i) > b$ as defined in \Cref{equ:upb1}. This implies that $\Lambda^2(S) \le \Lambda^0(S)$ and hence, $\Lambda^2$ is at least as good as $\Lambda^0$.

\begin{lemma}
\label{lem:maxvalid}
The maximum value $v$ that has a valid partition is given by $\Lambda^2(S)$. 
\end{lemma}
\begin{proof}[Proof Sketch]
Let the cost sequence generated for computing $\Lambda^2(S)$ be $s_1,s_2,\dots,s_n$ and the corresponding value sequence be $v_1,v_2,\dots,v_n$. If $s_1,s_2,\dots,s_n$ do not use up the budget $b$, 
the claim is straightforward since $\sum_{i=1}^{n} v_i = \sum_{i=1}^{n} f_{S \cup V_{i-1}}(i) = f_S(V_n)$ must be the maximum possible due to property (iv) of Definition \ref{def:partitionnew} (by setting $i=n$). Now assume that $\sum_{i=1}^{n}s_i = b$. We prove the lemma by contradiction. Let $\tilde{v}$ be the maximum value that has a valid partition and its cost sequence be $\tilde{s}_1,\tilde{s}_2,\dots,\tilde{s}_n$. Assume that $\tilde{v} > \Lambda^2(S)$. We check the difference between the two sequences and find the lowest-indexed element $a_q$ where $s_q \neq \tilde{s}_q$. We then derive contradictions for the cases of $\tilde{s}_q > s_q$ and $\tilde{s}_q < s_q$ respectively to complete the proof. Please refer to Appendix \ref{appendix_proof_of_lem:maxvalid} for details.
\end{proof}

\subsubsection{$\Lambda^{2}$ in linear program presentation}
\label{sec:lbd 2}
Like $\Lambda^1$, we can write the calculation of $\Lambda^2$ as a linear program shown in \Cref{lnp:lbd 2}.
\begin{align}
\label{lnp:lbd 2}
    \max & \quad \sum_{i \in V \setminus S} f_S(i) \cdot x_i\\
    \text{subject to } & \quad \forall i \in V, \quad \sum_{j = 1}^{i} f_S(j) \cdot x_j \le f_S(V_i), \nonumber \\
    & \quad \sum_{i \in V} c(i) \cdot x_i \le b. \nonumber 
\end{align}
    We prove in Appendix \ref{appendix_lp_lambda2} that the optimum of \Cref{lnp:lbd 2} is equal to 
    $\Lambda^2(S)$.

Though we have assumed $S=\{1,2,\dots,p\}$ and $V \setminus S=\{p+1,p+2,\dots,n\}$,
\Cref{lnp:lbd 2} actually does not require the elements of $S$ to be arranged before the elements of $V \setminus S$ in the ground set. The indexes of their elements can be arbitrarily interleaved. 
This is because in the first constraint of \Cref{lnp:lbd 2}, $f_S(j)=0$ for any element $j \in S$ and thus it would not increase the value of the left-hand side. On the right-hand side, $f_S(V_i)$ never decreases as $i$ increases. 
Thus, if the first constraint holds for all $i \in V \setminus S$, it would also hold for all $i \in S$. This 
observation will be useful when later we introduce a unified upper bound based on multiple base sets $S$, 
    where it is difficult to separate the indexes of $S$ and $V \setminus S$ for all base sets. 

\subsection{Combining removing and slicing strategies}
    We can combine the removing strategy and slicing strategy, bringing forth a tighter upper bound $\Lambda^3$. Given a set $S$, we denote the upper bound $\Lambda^2(S)$ generated with a budget $x \in [0, b]$ as $\Phi(x)$. The new upper bound $\Lambda^3$ is defined as follows:
\begin{equation*}
    \Lambda^3(S) = \max\limits_{x \in [0, c(S)]}(\Phi(b-c(S)+x)-G_-(x)),
\end{equation*}
    where $G_-(x)$ is defined in \Cref{eq:gminus}.

    According to the property of $\Lambda^2(S)$, 
    $\Phi(x)$ is a non-decreasing function and $\Phi(x) \ge \max\limits_{T \subseteq V \setminus S,\, c(T) \le x}f_S(T)$. 
    
    Letting $x = \max\{c(OPT \cup S) - b, 0 \}$, similar to the derivation in $\Lambda^1(S)$, we have
\begin{align*}
    & \Phi(b - c(S) + x) - G_-(x) \\
    & \ge  \Phi(c(OPT \cup S) - c(S)) 
    - G_-(c(OPT \cup S) - c(OPT)) \\ 
    & =  \Phi(c(OPT \setminus S)) - G_-(c(S\setminus OPT)) \\ 
    & \ge  f_S(OPT \setminus S) - f_{OPT}(S \setminus OPT) \\
    & =  f(OPT) - f(S),
\end{align*}
    which implies that $f(S) + \Lambda^3(S) \ge f(OPT)$.
    
    In addition, we have
\begin{align*}
    \Lambda^2(S) &= \Phi(b) 
    \ge \max\limits_{x \in [0, c(S)]}\Phi(b-c(S)+x) 
    \\ &
    \ge \max\limits_{x \in [0, c(S)]}(\Phi(b-c(S)+x)-G_-(x)) 
    = \Lambda^3(S),
\end{align*}
    which shows that $\Lambda^3$ is at least as good as $\Lambda^2$. 

The calculation of $\Lambda^3$ can also be written as a linear program by adding the removing item to \Cref{lnp:lbd 2}.
By similar arguments to the computational complexity of $\Lambda^1(S)$ in \Cref{sec:lbd 1}, the time complexity to compute $\Lambda^3(S)$ is also $O(n \log n) = O(|V| \log |V|)$, the same as $\Lambda^{0}(S)$, $\Lambda^{1}(S)$ and $\Lambda^{2}(S)$. Please refer to Appendix \ref{appendix_lbd_3_complexity} for details.

\subsection{Enhancing bounds with multiple base sets}
\label{sec:augment}
    If we calculate an upper bound with one base set $S$ only
    (e.g., an empty set or the output solution of an algorithm), 
    the result may not be satisfying. To tighten the bound, we can compute multiple upper bounds with different base sets (such as all the intermediate sets generated by a greedy algorithm for the MSMK problem) and choose the lowest upper bound obtained as the final result. 
    This enhancement can be conducted on all the proposed new upper bounds ($\Lambda^1$ to $\Lambda^3$). We denote the final result for such an enumeration based on $\Lambda^i$ as $\Lambda^{i+}$, i.e., $\Lambda^{i+} = \min_{S \in \{S_1, S_2, \dots\}} f(S) + \Lambda^i(S)$.
    
    In the above approach, one upper bound is computed for each base set separately. Alternatively, we can compute a shared upper bound for multiple base sets, 
    producing an even tighter bound.

\subsubsection{Unified augmentation for $\Lambda^0$ and $\Lambda^1$}

    We first construct a linear program as follows: 
\begin{align}
\label{lnp:lfomsm}
    \max & \quad z \\
    \text{subject to} & \quad \forall S \subseteq V, \quad z \le f(S) + \Lambda(S, \mathbf{x}), \nonumber \\
    & \quad \sum_{i \in V} c(i) \cdot x_i \le b, \nonumber
\end{align}
    where $\mathbf{x} = \langle x_1, x_2, \dots, x_n \rangle \in [0,1]^n$ is an indicator vector. Let $\mathbf{x^{OPT}}$ be the indicator vector of the optimal solution $OPT$ to the MSMK problem (where $x^{OPT}_i = 1$ if $i \in OPT$, and $x^{OPT}_i = 0$ otherwise). If function $\Lambda(S, \mathbf{x})$ satisfies $f(S) + \Lambda(S, \mathbf{x^{OPT}}) \ge f(OPT)$ for every set $S\subseteq V$, it is easy to see that the optimal value of \Cref{lnp:lfomsm} must be no less than $f(OPT)$. 

    \Cref{lnp:lfomsm} cannot be solved efficiently because there is an exponential number of possible sets $S$ (or constraints).    
    However, we can relax \Cref{lnp:lfomsm} by considering the constraints relevant to a collection of base sets $\{S_1, S_2, \dots, S_m\}$ only.
\begin{align}
\label{lnp:msmr}
    \max & \quad z \\
    \text{subject to} & \quad \forall i \in \{1,2,\dots,m\}, \quad z \le f(S_i) + \Lambda(S_i, \mathbf{x}), \nonumber \\
    & \quad \sum_{i \in V} c(i) \cdot x_i \le b. \nonumber
\end{align}
The optimal value of this linear program is not less than that of \Cref{lnp:lfomsm}, and thus guaranteed to be an upper bound on $f(OPT)$.
We call this procedure the ``unified'' augmentation method, since it forces each base set $S_i$ to share the same indicator vector $\mathbf{x}$.

For $\Lambda^0$ and $\Lambda^1$, we can design $\Lambda(S, \mathbf{x})$ according to the objective functions in their linear program representations (\Cref{lnp:lbd 0} and \Cref{lnp:lbd 1}). 

Specifically, we define $\Lambda^0(S, \mathbf{x})$ and $\Lambda^1(S, \mathbf{x})$ as follows:
\begin{equation}
    \Lambda^0(S, \mathbf{x}) = \sum_{i \in V \setminus S} f_S(i) \cdot x_i,
    \label{equ:lambda0}
\end{equation}
\begin{equation}
    \Lambda^1(S, \mathbf{x}) = \sum_{i \in V\setminus S} f_S(i) \cdot x_i - \sum_{i \in S} f_{V \setminus \{i\}}(i) \cdot (1-x_i).
\label{equ:lambda1}    
\end{equation}
Apparently, $\Lambda^0(S, \mathbf{x}) \ge \Lambda^1(S, \mathbf{x})$. For any set $S\subseteq V$, we have
\begin{eqnarray}
    \lefteqn{
    f(S) + \Lambda^1(S, \mathbf{x^{OPT}})
    } \nonumber \\ 
    &=& f(S) + \sum_{i \in V\setminus S} f_S(i) \cdot x^{OPT}_i - \sum_{i \in S} f_{V \setminus \{i\}}(i) \cdot (1-x^{OPT}_i) \nonumber \\
    &=& f(S) + \sum_{i \in OPT \setminus S} f_S(i) - \sum_{i \in S \setminus OPT} f_{V \setminus \{i\}} (i) \nonumber \\
    &\ge& f(S) + f_S(OPT) - f_{OPT}(S) 
    = f(OPT). \label{equ:lbd1}
\end{eqnarray}

Thus, we can instantiate $\Lambda(S, \mathbf{x})$ with $\Lambda^0(S, \mathbf{x})$ or $\Lambda^1(S, \mathbf{x})$ in \Cref{lnp:msmr}, and solve \Cref{lnp:msmr} to obtain an upper bound on $f(OPT)$. 
We denote the resulting bounds as $\Lambda^{0*}$ and $\Lambda^{1*}$.

\subsubsection{Unified augmentation for $\Lambda^2$ and $\Lambda^3$}
\label{sec:unified23}
To calculate a unified upper bound with multiple base sets $S_1, S_2, \dots, S_m$ for $\Lambda^2$, we construct the following linear program, where the second and third constraints correspond to conditions (iii) and (iv) of Definition \ref{def:partitionnew}, and $V_j := \{1,2,\dots,j\}$: 
\begin{align}
\label{lnp:unified lbd 2}
    \max & \quad z \\
    \text{subject to} & \quad \forall i \in \{1,2,\dots,m\}, \quad z \le f(S_i) +\sum_{j \in V} v^i_j, \nonumber \\
    & \quad \forall i \in \{1,2,\dots,m\} \text{ and } j \in V, \quad v^i_j \le f_{S_i}(j) \cdot x_j, \nonumber \\
    & \quad \forall i \in \{1,2,\dots,m\} \text{ and } j \in V, \quad \sum_{k=1}^{j} v^i_k \le f_{S_i}(V_j), \nonumber \\
    & \quad \sum_{i \in V} c(i) \cdot x_i \le b. \nonumber 
\end{align}

\Cref{lnp:unified lbd 2} can be viewed as an extension of \Cref{lnp:lbd 2} to multiple base sets (recall that \Cref{lnp:lbd 2} allows the indexes of the elements in $S$ and $V \setminus S$ to be interleaved arbitrarily).
It is easy to prove that the optimum of \Cref{lnp:unified lbd 2} dominates $f(OPT)$, since setting $\langle x_1, x_2, \dots, x_n \rangle = \mathbf{x^{OPT}}$ (the indicator vector of $OPT$) and $z = f(OPT)$ can satisfy all the constraints therein. 

To calculate a unified upper bound for $\Lambda^3$, we can further add the removing item to the first constraint of \Cref{lnp:unified lbd 2} as follows:
\begin{align}
\label{lnp:unified lbd 3}
    \max & \quad z \\
    \text{subject to} & \quad \forall i \in \{1,2,\dots,m\}, 
    \quad z \le f(S_i) +\sum_{j \in V} v^i_j - \sum_{j \in S_i} f_{V \setminus \{j\}}(j) \cdot (1-x_j), \nonumber \\
    & \quad \forall i \in \{1,2,\dots,m\} \text{ and } j \in V, \quad v^i_j \le f_{S_i}(j) \cdot x_j, \nonumber \\
    & \quad \forall i \in \{1,2,\dots,m\} \text{ and } j \in V, \quad \sum_{k=1}^{j} v^i_k \le f_{S_i}(V_j), \nonumber \\
    & \quad \sum_{i \in V} c(i) \cdot x_i \le b. \nonumber
\end{align}

\Cref{lnp:unified lbd 3} can be viewed as an extension of \Cref{lnp:lbd 3} to multiple base sets. 
Similarly, the optimum of \Cref{lnp:unified lbd 3} dominates $f(OPT)$, since setting $\langle x_1, x_2, \dots, x_n \rangle = \mathbf{x^{OPT}}$ 
and $z = f(OPT)$ can satisfy all the constraints therein. 

We denote the upper bounds obtained from Equations \eqref{lnp:unified lbd 2} and \eqref{lnp:unified lbd 3} as $\Lambda^{2*}$ and $\Lambda^{3*}$ respectively.

\section{Experiments}
\label{sec:exp}
    We carry out experiments on three different applications to demonstrate the advantage of our proposed bounds ($\Lambda^{1}$, $\Lambda^{2}$, $\Lambda^{3}$ and their augmentations) over the previous bounds $\Lambda^{0}$ and $\Lambda^{0+}$ in \cite{tang2021revisiting}.
    The experiments are conducted on a Windows machine with a 3.8GHz Intel Xeon W-2235 CPU and 32GB RAM with code written in Python.

\noindent \textbf{Feature selection.} Feature selection is a procedure widely used in machine learning~\cite{amiridi2021information,bao2022submodular,das2011submodular}. Given a feature set $V$ which contains all features that a data point in a training set may have, the goal is to select a subset $S \subseteq V$, with the highest utility subject to a limited budget, to construct a classification model. We adopt the following entropy function~\cite{balkanski2021instance} to determine the utility value of a feature set $S \subseteq V$: 
\begin{equation*}
    f(S) = -\sum_{x \in X_S}\sum_{y \in Y}p(x, y)\log p(x,y),
\end{equation*}
where $X_S$ is the feature matrix indexed by $S$, 
$Y$ is the set of labels, and $p(x,y)$ is the joint distribution on $(X_S,Y)$.
We experiment with \textbf{Adult Income}, a real-world dataset from~\cite{misc_adult_2}. This dataset contains 32561 individuals with a variety of features. We retrieve 111 binary features from the dataset as in ~\cite{kazemi2018scalable} and then randomly generate 20 ground sets each containing 100 features 
by uniformly sampling from these 111 binary features.

\noindent \textbf{Maximum coverage.} Maximum coverage is a classical optimization problem with numerous practical applications, such as influence maximization~\cite{tang2018online,kempe2003maximizing} and sensor placement~\cite{krause2008near,leskovec2007cost}. It addresses the challenge of using limited resources to achieve the highest utility. 
Given a graph $G=(V, E)$ where each vertex $v \in V$ has an associated cost $c(v)$, the objective is to find a vertex subset $S \subseteq V$ with limited total cost that maximizes a coverage function:  
\begin{equation*}
    f(S) = |N_G(S)|,
\end{equation*}
where $N_G(S)$ is the graph neighborhood function that returns a set containing all vertices in $S$ and all vertices adjacent to any vertex in $S$ in the graph $G$. 
We experiment with \textbf{ego-facebook} and \textbf{com-youtube}, two real-world datasets from~\cite{snapnets}. \textbf{ego-facebook} contains 4039 nodes and 88234 edges; \textbf{com-youtube} (top 5000 communities) contains 39841 nodes and 224234 edges. 
For each dataset, we randomly construct 20 ground sets each containing 1000 nodes by uniformly sampling from the dataset.

\noindent \textbf{Revenue maximization.} In social advertising, we would like to choose seed users who will advertise products to their neighbors on social networks. A social network is modeled by a graph $G=(V,E)$, where each node $v \in V$ represents a user, and each edge $(u,v) \in E$ has a weight $w_{uv}$ describing the influence of $u$ on $v$. Each user $v$ will be paid a cost $c(v)$ for advertising products if chosen as a seed user. The goal is to select a subset $S \subseteq V$ of seed users within a budget to maximize product revenue.
We adopt the expected product revenue function defined in ~\cite{breuer2020fast}:
\begin{equation*}
    f(S) = \sum_{v\in V}\Big(\sum_{u \in S}w_{uv}\Big)^{0.9},
\end{equation*}
where 
the exponent $0.9$ measures diminishing returns.
We experiment with a real-world dataset \textbf{Caltech36} from~\cite{nr} containing 769 users. We randomly generate 20 ground sets each containing 100 nodes by uniformly sampling from the dataset.

\noindent \textbf{Cost and budget setting.}
All the datasets 
lack a cost function. When conducting the experiments for each ground set, we randomly assign a cost to each element 
in three different ways: (a) all element costs are generated from $U(1,5)$, where $U(x,y)$ represents a uniform distribution between $x$ and $y$; 
(b) 80\% element costs are generated from $U(1,5)$ and 20\% element costs are generated from $U(5,20)$ (most costs are small while some are large);
(c) 80\% element costs are generated from $U(1,5)$ and 20\% element costs are generated from $U(0,1)$ (most costs are large while some are small).
Similar performance trends are observed for these three distributions.
We present the results for (a) here and defer the results for (b) and (c) to Appendix \ref{sec:add_cost_setting}.
For cost setting (a), we conduct experiments for all ground sets generated with different budgets from 6 to 40 (step size 1). 

\noindent \textbf{Algorithms and intermediate sets.}
For each problem instance, we run 
the modified greedy algorithm \textbf{MGreedy} 
\cite{wolsey1982maximising} and the \textbf{Greedy+Max} algorithm \cite{yaroslavtsev2020bring}.
The experimental results for the two algorithms show similar trends. Due to space limitations, we focus on presenting the results of \textbf{MGreedy} in this paper.
\textbf{MGreedy} starts with an empty solution set and employs a greedy heuristic.
In each iteration, 
the greedy heuristic finds the element with the highest marginal density among all elements that can fit into the remaining budget, and adds it to the solution set. 
The procedure repeats until no more element can be added to the solution set $S$ due to budget violation. \textbf{MGreedy} then finds an element $v^*$ which maximizes $f(\{v^*\})$ and compares it with $f(S)$ and outputs the one with a larger function value.
The approximation factor of \textbf{MGreedy} is between 0.427 and 0.42945~\cite{Feldman2023,Kulik2021}.

We calculate the upper bounds $\Lambda^{0}$, $\Lambda^{1}$, $\Lambda^{2}$, $\Lambda^{3}$ using the empty set $\emptyset$ or the algorithm's output solution $S$ as the base set, and their enumerated and unified augmentations $\Lambda^{i+}$ and $\Lambda^{i*}$
based on all intermediate sets generated by the algorithm 
(i.e., those produced by the greedy heuristic 
as well as the singleton set $\{v^*\}$). 
$\Lambda^i$ and $\Lambda^{i+}$ are calculated by the formulas 
presented. $\Lambda^{i*}$ is calculated by solving the linear program presented, for which we use 
the SciPy library in Python. 
Then, we divide the \textbf{MGreedy} solution by the upper bound to derive its approximation guarantee.
For each dataset, there are 700 problem instances (20 ground sets $\times$ 35 different budgets from 6 to 40). 
We present the box plot to illustrate the experimental results across all instances.
We use a common box plot showing the first quartile (25th percentile) and the third quartile (75th percentile) as a box, and the boundaries of the whiskers based on 1.5 times the interquartile range 
(the distance between the third and first quartiles).
In addition, the green triangle presents the mean value. 

\begin{figure}[!t]

\centering
\includegraphics[
width=0.7\textwidth]{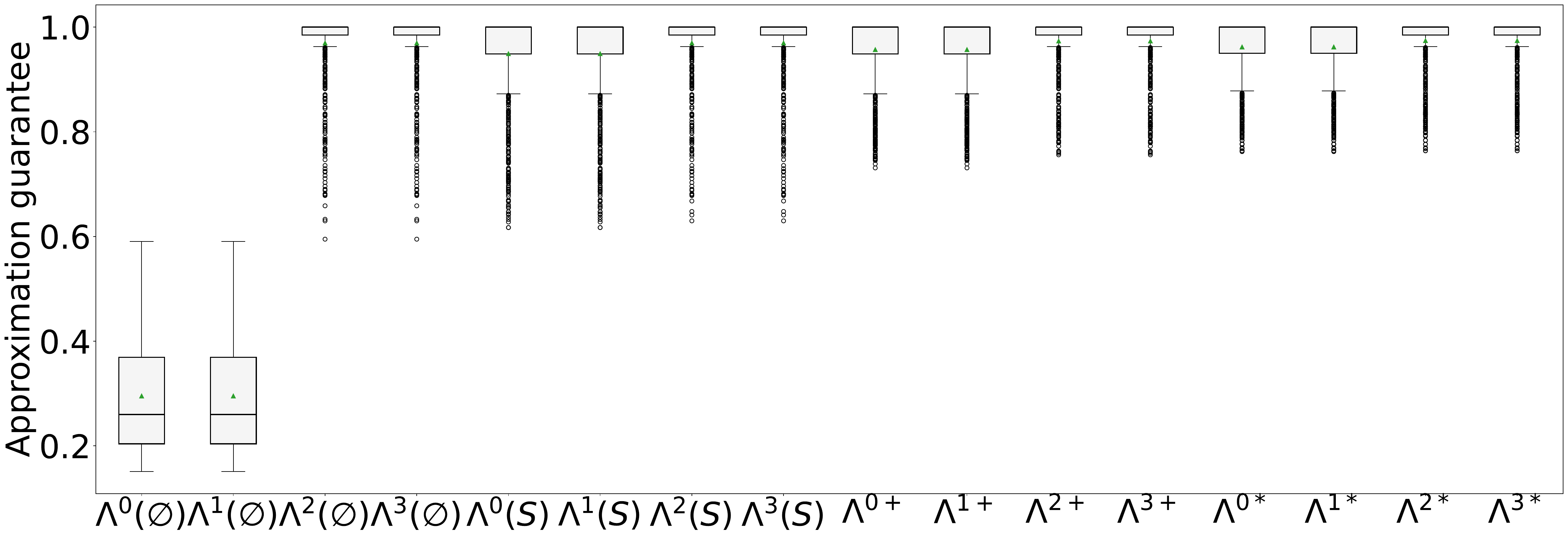} \\
(a) Adult Income 

\centering
\includegraphics[
width=0.7\textwidth]{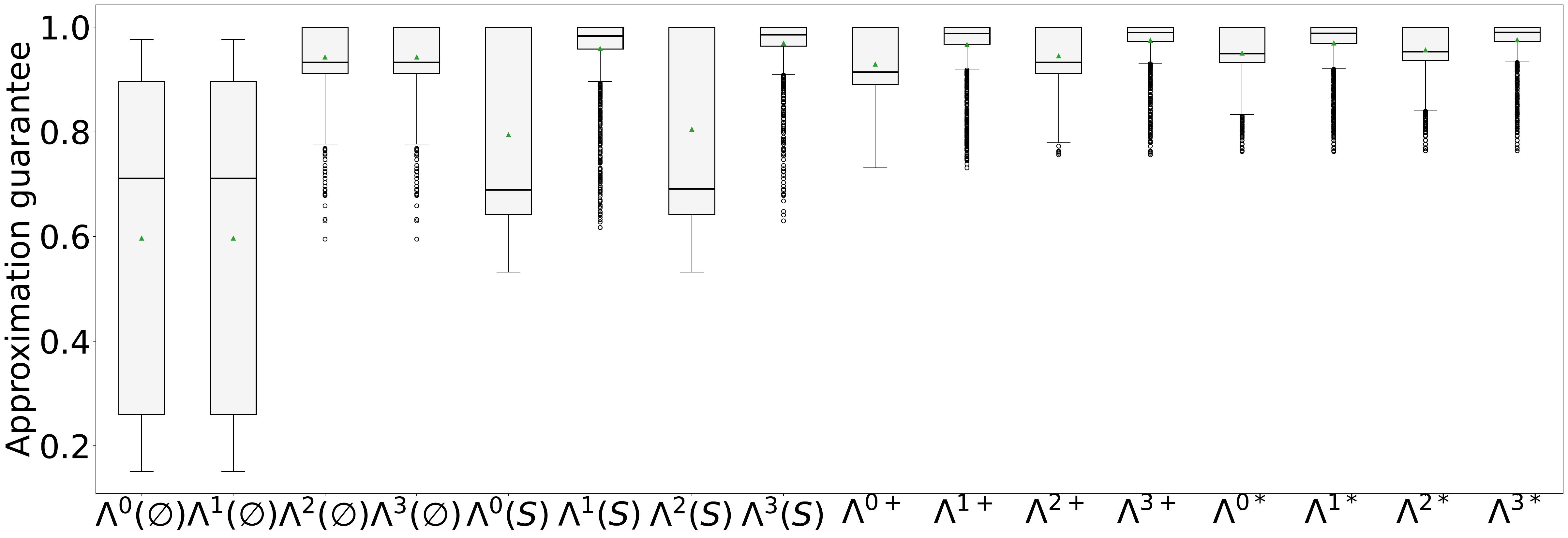} \\
(b) Caltech36 

\centering
\includegraphics[
width=0.7\textwidth]{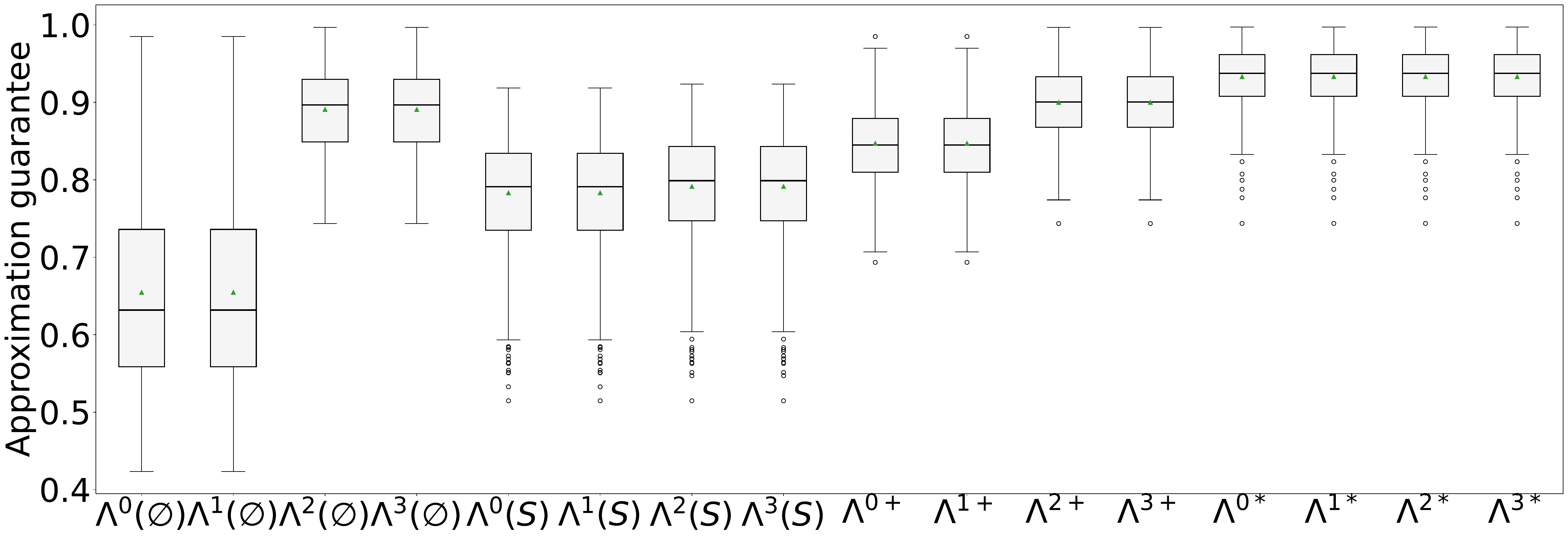} \\
(c) ego-facebook 

\centering
\includegraphics[
width=0.7\textwidth]{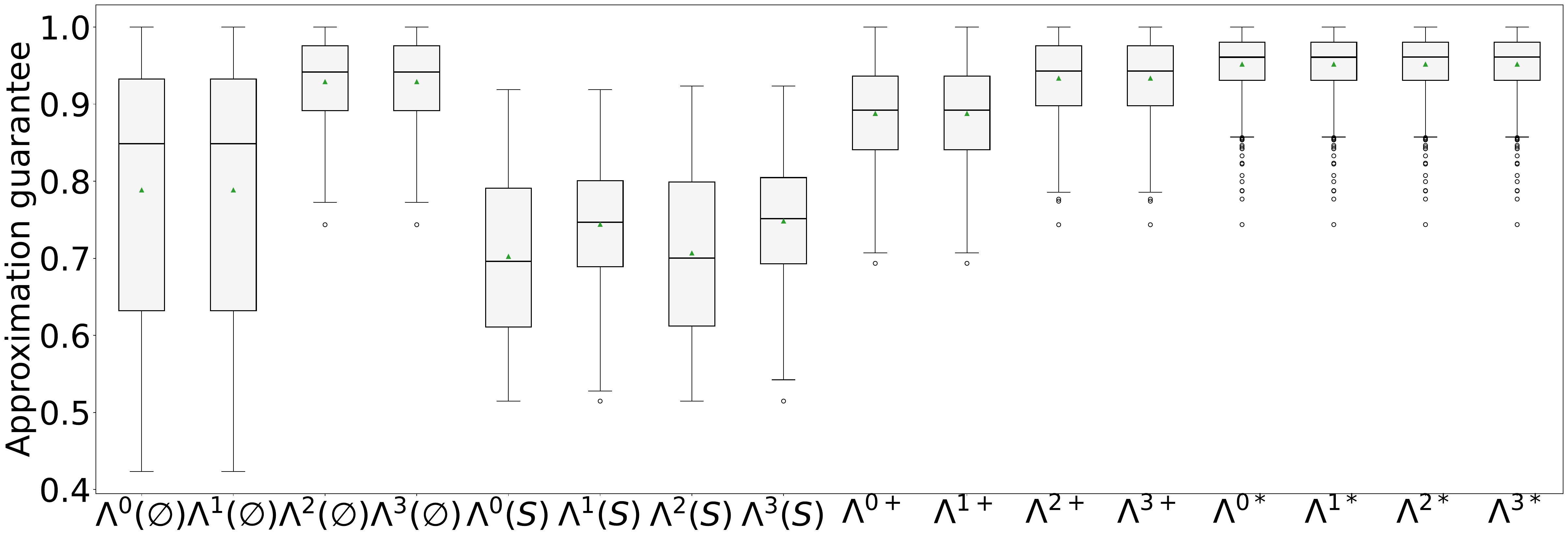} \\
(d) com-youtube 

\caption{Actual approximation guarantee plots for different upper bounds} 
\label{fig:frl}
\end{figure}

\Cref{fig:frl} presents the approximation guarantees derived from different upper bounds for the \textbf{MGreedy} algorithm.
We can make the following observations from the results.
First, data-dependent upper bounds typically certify much higher approximation guarantees than the worst-case theoretical result (recall that \textbf{MGreedy}'s approximation factor is between 0.427 and 0.42945). The unified upper bounds $\Lambda^{i*}$ confirm that \textbf{MGreedy} solutions are quite close to optimal (within 10\% in most cases).
Second, $\Lambda^{2}/\Lambda^{2+}/\Lambda^{2*}$ are significantly tighter than $\Lambda^{0}/\Lambda^{0+}/\Lambda^{0*}$ for most problem instances, which demonstrates the advantage of the slicing strategy.
$\Lambda^{3}/\Lambda^{3+}/\Lambda^{3*}$ (resp. $\Lambda^{1}/\Lambda^{1+}/\Lambda^{1*}$) are at least as good as $\Lambda^{2}/\Lambda^{2+}/\Lambda^{2*}$ (resp. $\Lambda^{0}/\Lambda^{0+}/\Lambda^{0*}$) as proved. As seen from \Cref{fig:frl}(b), the augmentations $\Lambda^{3+}/\Lambda^{3*}$ (resp. $\Lambda^{1+}/\Lambda^{1*}$) are much tighter than $\Lambda^{2+}/\Lambda^{2*}$ (resp. $\Lambda^{0+}/\Lambda^{0*}$) for the \textbf{Caltech36} dataset, which will be further analyzed below. 
Third, upper bounds calculated with only one base set $\emptyset$ or $S$ can be rather loose (e.g., $\Lambda^0(\emptyset)$ and $\Lambda^1(\emptyset)$ for the \textbf{Adult Income} dataset). The relative performance of $\Lambda^i(\emptyset)$ and $\Lambda^i(S)$ varies with the dataset. Upper bounds calculated with multiple base sets are normally considerably tighter. The unified upper bounds $\Lambda^{i*}$ are remarkably tighter than the enumerated upper bounds $\Lambda^{i+}$ in many cases. 

We present two reasons accounting for the superior performance of the removing strategy
on the \textbf{Caltech36} dataset.

First, an element $e$ in the \textbf{Caltech36} dataset is more likely to have a higher cutoff density $d_{V \setminus \{e\}}(e)$
compared with the other three datasets, which makes the removing strategy more effective. In \Cref{fig:tmm2 10}, we present the $d_{V \setminus \{e\}}(e)$ values of elements from different ground sets. The plot contains data from 20000/20000/2000/2000 elements from the 20 ground sets of the \textbf{ego-facebook}/\textbf{com-youtube}/\textbf{Caltech36}/\textbf{Adult Income} datasets, respectively. 

\begin{figure}[!t]
    \centering
        \includegraphics[width=0.5\linewidth]{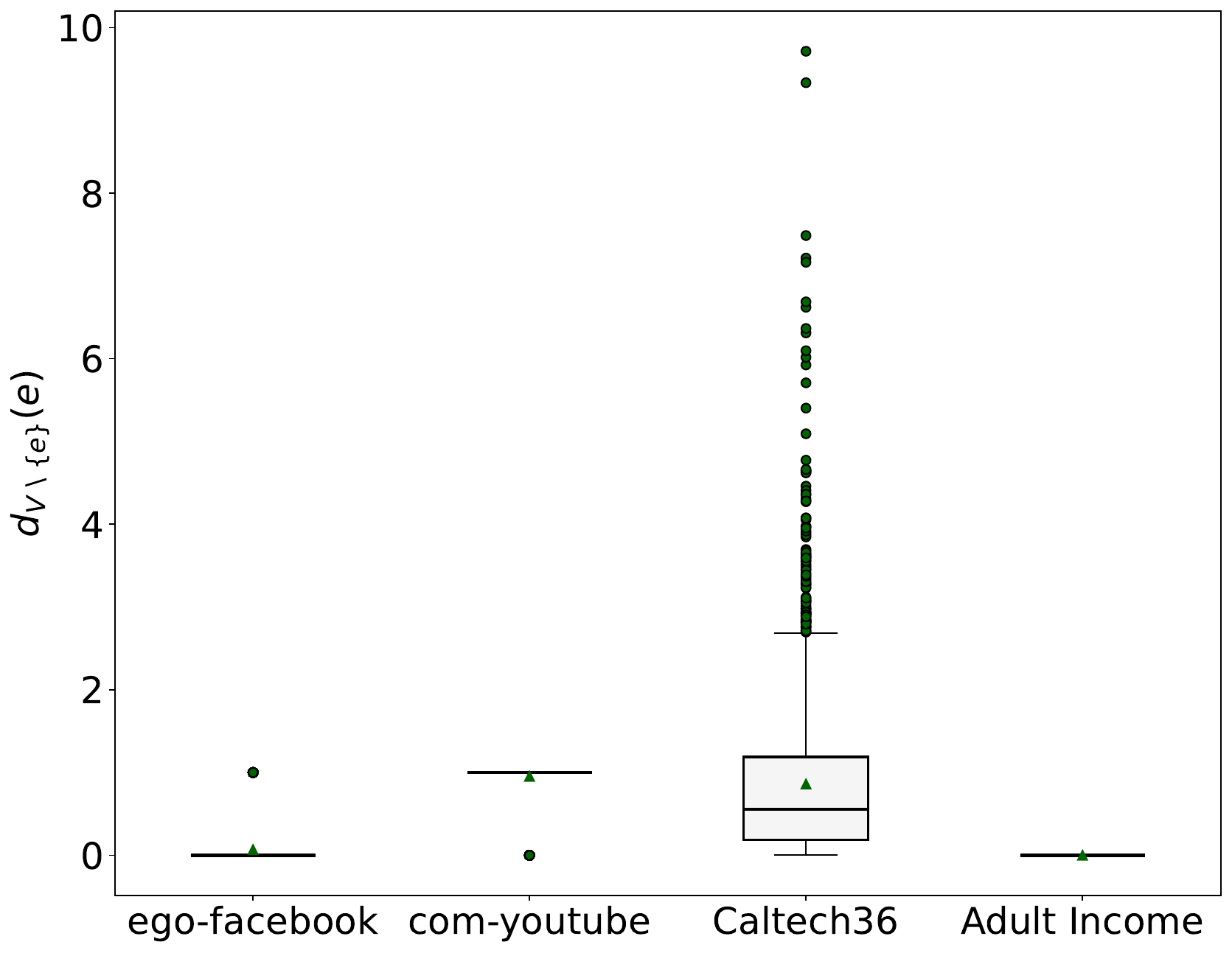}
        \caption{Cutoff densities of elements in four datasets}
        \label{fig:tmm2 10}
\end{figure}

Second, recall that the enumerated augmentation $\Lambda^{i+}$ is the lowest upper bound among all those obtained from taking the intermediate sets of \textbf{MGreedy} as base sets. If $\Lambda^{1+}$ or $\Lambda^{3+}$ is given by the upper bound obtained from an empty set (the initial greedy solution), the removing strategy does not further tighten the bound because
\begin{align*}
    \Lambda^{1}(\emptyset)
    &= \max\limits_{x \in [0, c(\emptyset)]}\left\{G_+(b-c(\emptyset)+x) - G_-(x)\right\} \\ 
    &= G_+(b-c(\emptyset)+0) - G_-(0) \\
    & =  G_+(b) \\
    & =  \Lambda^{0}(\emptyset),
\end{align*}
and a similar argument also holds for $\Lambda^{3+}$. 
$\Lambda^{3+}$ is never equal to $\Lambda^{3}(\emptyset)$ for the \textbf{Caltech36} dataset, while it is common for the other three datasets as  illustrated in \Cref{tab:empty set count}. 
This observation explains the reason why the \textbf{com-youtube} dataset fails to have a better $\Lambda^{3+}$ than $\Lambda^{2+}$ despite its 
cutoff densities $d_{V \setminus \{e\}}(e)$ not being close to 0. 
For similar reasons, the unified augmentations $\Lambda^{1*}/\Lambda^{3*}$ have the same performance tendencies as the enumerated augmentations $\Lambda^{1+}/\Lambda^{3+}$.  

\begin{table}[htbp]
\caption{Problem instances where $\Lambda^{3+}=\Lambda^{3}(\emptyset)$}
\label{tab:empty set count}
\begin{center}
\begin{tabular}{c|c|c|c}
\hline
Dataset & \# of $\Lambda^{3+}=\Lambda^{3}(\emptyset)$ & Total \# & Percentage \\
\hline
Adult Income & 658 & 700 & 94.0\% \\
\hline
Caltech36 & 0 & 700 & 0.0\% \\
\hline
ego-facebook & 396 & 700 & 56.8\% \\
\hline
com-youtube & 700 & 700 & 100\% \\
\hline
\end{tabular}
\end{center}
\end{table}

To evaluate the computational efficiency of our upper bounds, we present the box plots of computational time 
in \Cref{fig:T}.
The computational times of $\Lambda^{2}/\Lambda^{2+}$ and $\Lambda^{0}/\Lambda^{0+}$ are nearly the same for all the datasets except \textbf{Adult Income}.\footnote{Detailed examinations show that the number of elements visited in $V \setminus S$ for computing $\Lambda^{2}/\Lambda^{2+}$ is much larger than that for computing $\Lambda^{0}/\Lambda^{0+}$ for the \textbf{Adult Income} dataset, while these numbers are similar for other datasets.} This demonstrates the computational efficiency of the slicing strategy. 
The computational times of $\Lambda^{1}/\Lambda^{1+}$ and $\Lambda^{0}/\Lambda^{0+}$ are also generally close for all the datasets. This shows that the removing strategy does not introduce much additional computational overhead.
Comparing enumerated and unified augmentations, the computational times of $\Lambda^{0*}/\Lambda^{1*}$ are similar to $\Lambda^{0+}/\Lambda^{1+}$, whereas the computational times of $\Lambda^{2*}/\Lambda^{3*}$ can be much higher than $\Lambda^{2+}/\Lambda^{3+}$. This is because the number of constraints is independent of the ground set size in the linear program of unified augmentations $\Lambda^{0*}/\Lambda^{1*}$, while it is proportional to the ground set size in the linear program of unified augmentations $\Lambda^{2*}/\Lambda^{3*}$. 

\begin{figure}[!t]

\centering
\includegraphics[
width=0.7\textwidth]{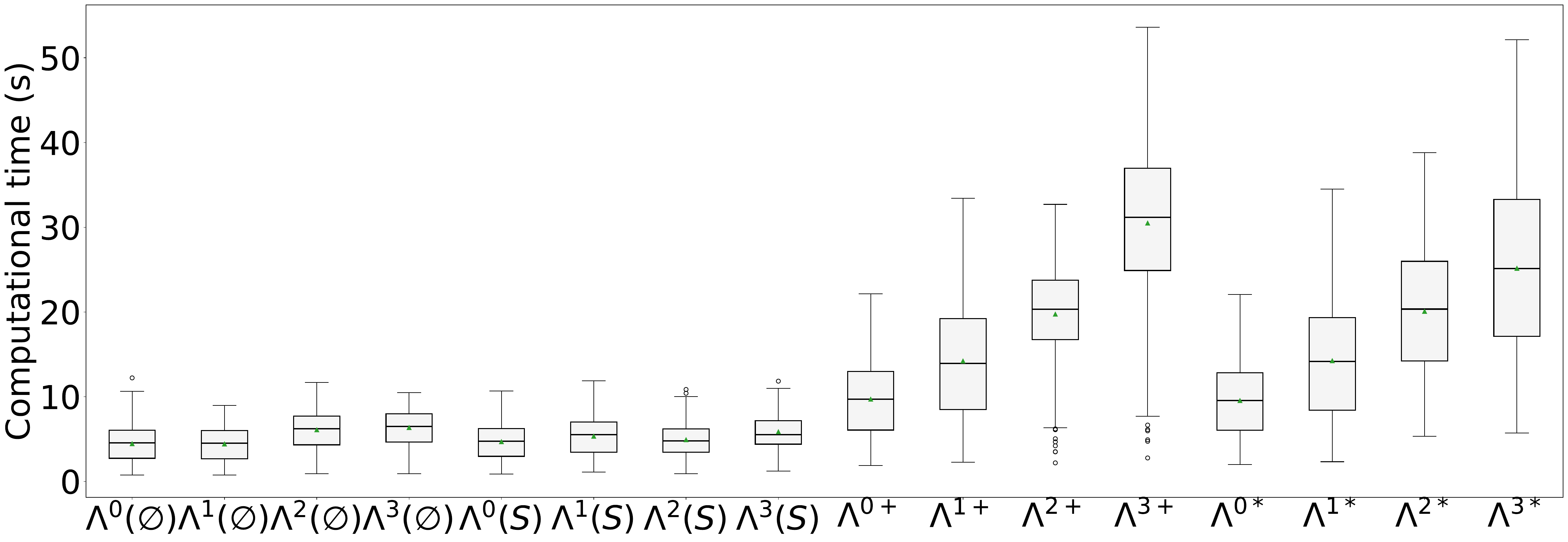} \\
(a) Adult Income 

\centering
\includegraphics[
width=0.7\textwidth]{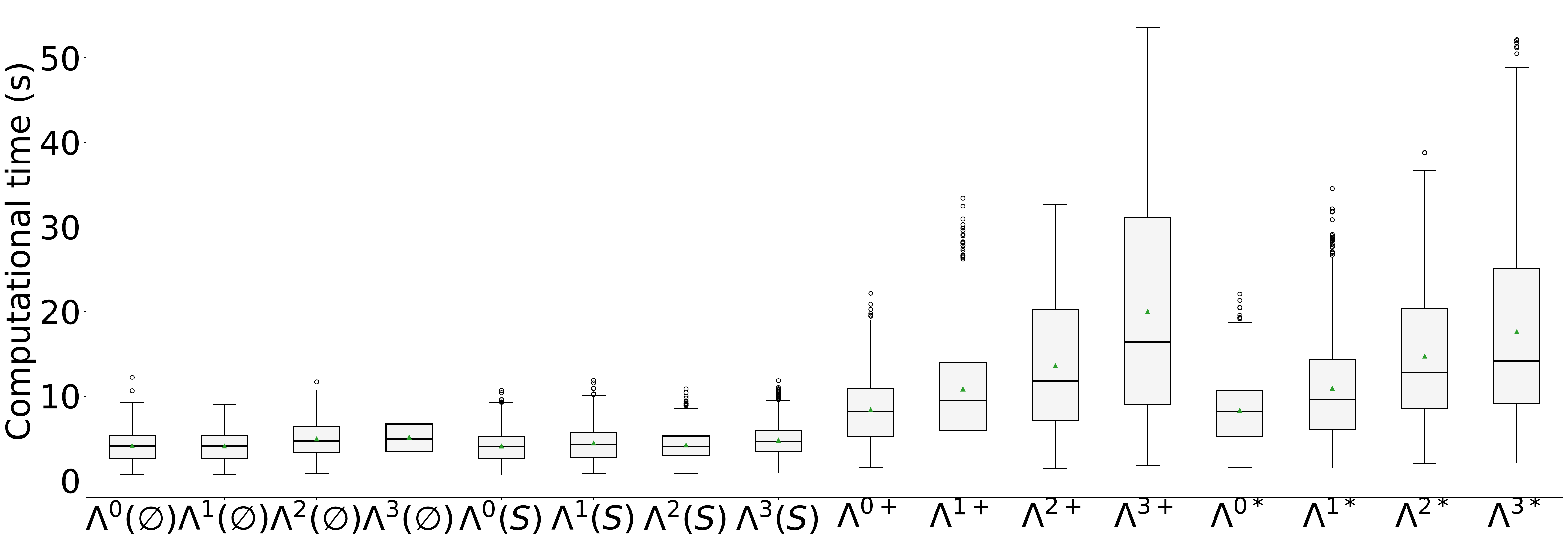} \\
(b) Caltech36 

\centering
\includegraphics[
width=0.7\textwidth]{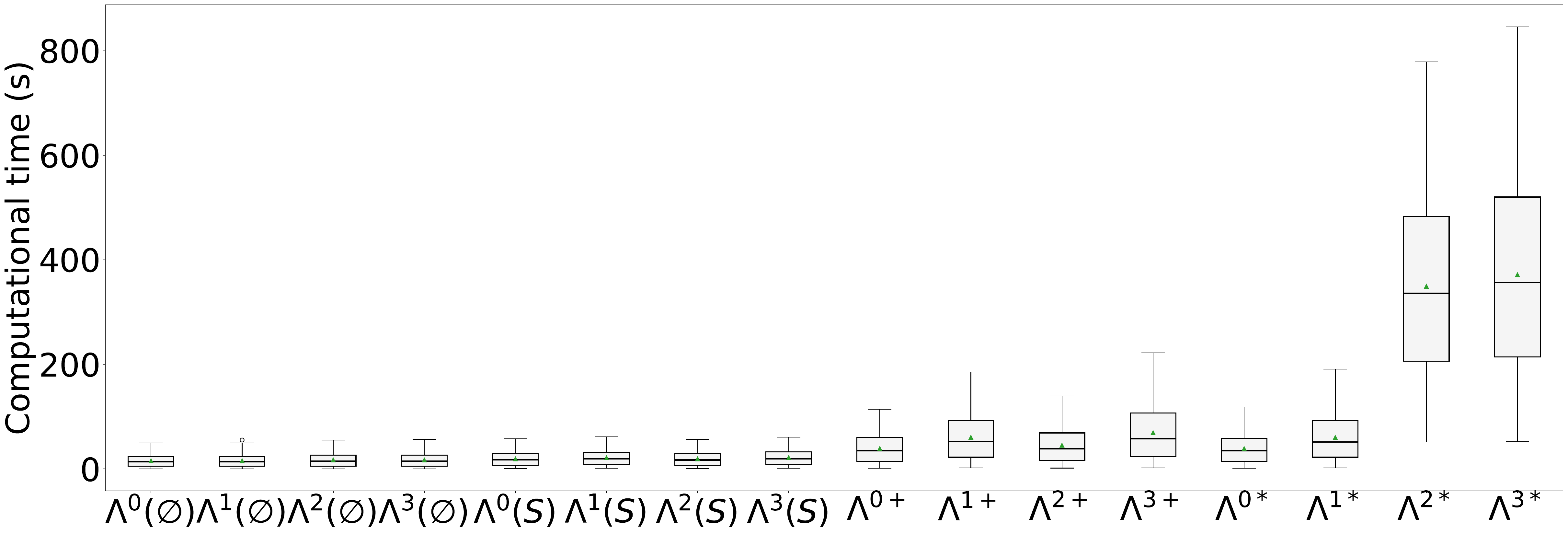} \\
(c) ego-facebook 

\centering
\includegraphics[
width=0.7\textwidth]{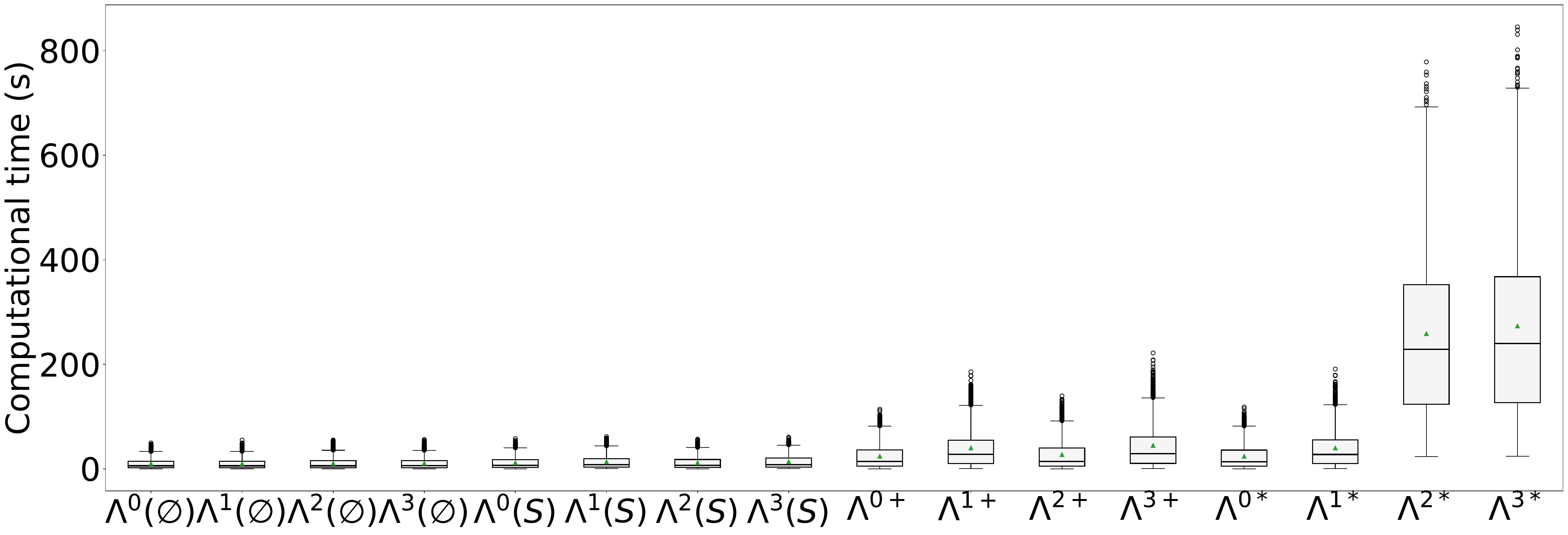} \\
(d) com-youtube 

\caption{Computational time plots for different upper bounds} 
\label{fig:T}
\end{figure}

For large-scale problems where computational time may be a concern, we recommend using the enumerated augmentation $\Lambda^{3+}$ and the unified augmentation $\Lambda^{1*}$. They often have much lower computational times than the unified augmentations $\Lambda^{2*}$ and $\Lambda^{3*}$ (see \textbf{ego-facebook} and \textbf{com-youtube} results in \Cref{fig:T}, where the ground set size is 1000), while the better bound between $\Lambda^{3+}$ and $\Lambda^{1*}$ is generally close to $\Lambda^{2*}$ and $\Lambda^{3*}$ as shown in \Cref{fig:frl}. 

\section{Conclusion}
\label{sec:con}

    In this paper, we have proposed a removing strategy and a slicing strategy for constructing new data-dependent upper bounds for submodular maximization with a knapsack constraint. 
    The upper bounds constructed are shown to be much tighter than existing bounds
    by extensive experiments. 
In future work, we would like to find the ways to extend our techniques to submodular maximization problems with other constraints, such as matroid and $p$-system. 

\section*{Acknowledgments}
This research is supported by the Ministry of Education, Singapore, under its Academic Research Fund Tier 2 (Award MOE-T2EP20122-0007). Jing Tang’s work is also partially supported by National Key R\&D Program of China under Grant No.\ 2024YFA1012700, by the National Natural Science Foundation of China (NSFC) under Grant No.\ 62402410, and by Guangdong Provincial Project (No.\ 2023QN10X025).

\appendix

\section{Equivalence between the optimum of Equation (\ref{lnp:lbd 1}) and $\Lambda^1(S)$}
\label{appendix_lp_lambda1}
\begin{lemma}
\label{lem:lnp lbd 1}
    For each value $t \in [0, c(S)]$, there exists an indicator vector $\mathbf{x} = \langle x_1, x_2, \dots, x_n \rangle \in [0,1]^n$ that satisfies $G_+(b-c(S)+t)-G_-(t)=\sum_{i \in V \setminus S} f_S(i) \cdot x_i - \sum_{i \in S} f_{V \setminus \{i\}}(i) \cdot(1-x_i)$ and $\sum_{i \in V} c(i) \cdot x_i \le b$.
\end{lemma}
\begin{proof}
    We prove this lemma by constructing an $\mathbf{x}$ that satisfies the constraints.

    Let $r$ be the index of the last element covered by $G_+(b-c(S)+t)$, i.e., $\sum_{i=p+1}^{r-1}f_S(i) < G_+(b-c(S)+t) \le \sum_{i=p+1}^{r}f_S(i)$. Similarly, let $s$ be the index of the last element covered by $G_-(t)$, i.e., $\sum_{i=1}^{s-1}f_{V\setminus\{i\}}(i) < G_-(t) \le \sum_{i=1}^{s}f_{V\setminus\{i\}}(i)$.

    We construct an $\mathbf{x}$ as follows: 
\begin{equation*}
x_i=\left\{
\begin{aligned}
    0 &,& 1 \le i \le s-1, \\
    \frac{\sum_{i=1}^{s}c(i)-t}{c(s)} &,& i =s,\\
    1 &,& s+1 \le i \le p, \\
    1 &,& p+1 \le i \le r-1, \\
    \frac{b-c(S)+t-\sum_{i=p+1}^{r-1}c(i)}{c(r)}&,& i=r, \\
    0 &,& r+1 \le i \le n.
\end{aligned}
\right.
\end{equation*}

Then, we have
\begin{eqnarray*}
    G_+(b-c(S)+t) &=& \sum_{i=1}^{r-1}f_S(i) 
    +\left(b-c(S)+t-\sum_{i=1}^{r-1}c(i)\right) \cdot \frac{f_S(r)}{c(r)} \\
    &=& \sum_{i=p+1}^{r-1}f_S(i) \cdot x_i + f_S(r) \cdot x_r 
    = \sum_{i \in V \setminus S}f_S(i) \cdot x_i,
\end{eqnarray*}
and
\begin{align*}
    & G_-(t) = \sum_{i=1}^{s-1} f_{V \setminus \{i\}}(i) + \left(t-\sum_{i=1}^{s-1}c(i)\right) \cdot \frac{f_{V \setminus \{s\}}(s)}{c(s)} \\
    &= \sum_{i=1}^{s-1} f_{V \setminus \{i\}}(i) \cdot (1-x_i) + f_{V \setminus \{s\}}(s) \cdot (1-x_s) 
    = \sum_{i \in S} f_{V \setminus \{i\}}(i) \cdot (1-x_i).
\end{align*}
Thus, $G_+(b-c(S)+t)-G_-(t)=\sum_{i \in V \setminus S}f_S(i) \cdot x_i -\sum_{i \in S}f_{V \setminus \{i\}}(i) \cdot (1-x_i)$.

Finally,
\begin{equation*}
    \sum_{i \in V} c(i) \cdot x_i 
    = \sum_{i \in S} c(i) \cdot x_i + \sum_{i \in V\setminus S} c(i) \cdot x_i 
    = c(S) - t + b-c(S) + t 
    = b.
\end{equation*}    
Hence, the $\mathbf{x}$ constructed satisfies the two constraints. 
\end{proof}

\begin{lemma}
\label{lem:lnp lbd 1 part 2}
    For each indicator vector $\mathbf{x}  = \langle x_1, x_2, \dots, x_n \rangle \in [0,1]^n$ that satisfies $\sum_{i\in V}c(i) \cdot x_i \le b$, there exists a value $t \in [0, c(S)]$ such that $G_+(b-c(S)+t)-G_-(t) \ge \sum_{i \in V \setminus S} f_S(i) \cdot x_i - \sum_{i \in S} f_{V \setminus \{i\}}(i) \cdot(1-x_i)$.
\end{lemma}
\begin{proof}
    It suffices to prove the claim for an indicator vector producing the optimal function value of Equation (\ref{lnp:lbd 1}), which maximizes the right-hand side of the inequality to be proved. Since the coefficients of all variables $x_i$'s are non-negative, the objective function value of Equation (\ref{lnp:lbd 1}) never decreases with increasing $x_i$ values. Thus, without loss of generality, we can assume that the indicator vector producing the optimal function value of Equation (\ref{lnp:lbd 1}) satisfies $\sum_{i\in V}c(i) \cdot x_i = b$. Let $t = \sum_{i \in V \setminus S} c(i) \cdot x_i - b + c(S) = c(S) - \sum_{i \in S} c(i) \cdot x_i \ge 0$. 
    We can infer that $\sum_{i \in V\setminus S} f_S(i) \cdot x_i \le G_+(\sum_{i \in V\setminus S}c(i) \cdot x_i)$ since $G_+$ greedily picks elements with the highest marginal densities. Thus,    
\begin{eqnarray*}
    \sum_{i \in V \setminus S} f_S(i) \cdot x_i \le G_+(b-c(S)+t).
\end{eqnarray*}
    Similarly, we can infer that $\sum_{i \in S} f_{V \setminus \{i\}}(i) \cdot (1 - x_i) \ge G_-(\sum_{i \in S}c(i)\cdot (1-x_i))$ since $G_-$ greedily picks elements with the lowest cutoff densities. Moreover, it follows from $\sum_{i \in V}c(i) \cdot x_i = b$ that $\sum_{i \in S}c(i)\cdot (1-x_i) = c(S) - \sum_{i \in S}c(i)\cdot x_i = c(S)-b+\sum_{i \in V \setminus S} c(i) \cdot x_i$. Therefore,
\begin{align*}    
    \sum_{i \in S} f_{V \setminus \{i\}}(i) \cdot (1 - x_i) & \ge G_-(\sum_{i \in S}c(i)\cdot (1-x_i)) 
    \\ & 
    = G_-(c(S)-b+\sum_{i \in V \setminus S} c(i) \cdot x_i) 
    = G_-(t).
\end{align*}

    Hence, we can conclude that $G_+(b-c(S)+t)-G_-(t) \ge \sum_{i \in V \setminus S} f_S(i) \cdot x_i - \sum_{i \in S} f_{V \setminus \{i\}}(i) \cdot(1-x_i)$.
\end{proof}

    Lemma \ref{lem:lnp lbd 1} and Lemma \ref{lem:lnp lbd 1 part 2} together guarantee the equivalence between the optimum of \Cref{lnp:lbd 1} and $\Lambda^1(S)$.

\section{Proof of Lemma \ref{lem:uboffT}}\label{appendix_proof_of_lem:uboffT}
\textbf{Lemma \ref{lem:uboffT}.}
\textit{For any set $T \subseteq V \setminus S$ with $c(T) \le b$, $f_S(T)$ has a valid partition with respect to $S$.}
\begin{proof} 
We construct a valid partition of $f_S(T)$ to prove this lemma. 
Assume that $T$ has $m$ elements: 
$T=\{t_1,t_2,\dots,t_m\}$ 
where $p+1 \leq t_1 < t_2 < \cdots < t_m \leq n$. To simplify presentation, we define $T_i:=\{t_1,t_2,\dots,t_i\}$ and $T_0:= \emptyset$. 

For each element $t_i \in T$, we set $s_{t_i} = \frac{d_{S \cup T_{i-1}}(t_i)}{d_S(t_i)}\cdot c(t_i)$ and $v_{t_i} = f_{S \cup T_{i-1}}(t_i)$.
For each element $j \in V \setminus T$, we set $s_j = 0$ and $v_j = 0$.
We verify all the conditions of Definition \ref{def:partitionnew}.

For condition (i), 
\begin{equation*}
\sum_{i=1}^n v_i = \sum_{i=1}^m v_{t_i} = \sum_{i=1}^m f_{S \cup T_{i-1}}(t_i) = f_S(T).
\end{equation*}

For condition (ii), 
\begin{align*}
\sum_{i=1}^n s_i &= \sum_{i=1}^m s_{t_i} 
 = \sum_{i=1}^m \frac{d_{S \cup T_{i-1}}(t_i)}{d_S(t_i)}\cdot c(t_i) 
 \le \sum_{i=1}^m c(t_i) = c(T) \le b,
\end{align*}
where the first inequality is due to the submodularity of $f$. 

For condition (iii),
for each element $t_i \in T$, 
\begin{equation*}
d_S(t_i) \cdot s_{t_i} = d_{S \cup T_{i-1}}(t_i) \cdot c(t_i) = f_{S \cup T_{i-1}}(t_i) = v_{t_i},
\end{equation*}
and for each element $j \in V \setminus T$,
\begin{equation*}
d_S(j) \cdot s_j = 0 = v_j.
\end{equation*}

For condition (iv),
for each element $j \in V\setminus T$, let $r$ be the highest index 
satisfying $t_r \le j$ (define $r = 0$ and $t_r = 0$ if $j < t_1$), then
\begin{equation*}
f_S(V_{j}) \ge \sum_{k=1}^{r} f_{S \cup T_{k-1}}(t_k) 
= \sum_{k=1}^{r} v_{t_k} = \sum_{k=1}^{t_r} v_k 
= \sum_{k=1}^{j} v_k,
\end{equation*}
where the inequality is due to the monotonicity of $f$. 
\end{proof}

\section{Proof of Lemma \ref{lem:maxvalid}}\label{appendix_proof_of_lem:maxvalid}
\textbf{Lemma \ref{lem:maxvalid}.}
\textit{The maximum value $v$ that has a valid partition is given by $\Lambda^2(S)$.} 
\begin{proof}
Let the cost sequence generated for computing $\Lambda^2(S)$ be $s_1,s_2,\dots,s_n$ and the corresponding value sequence be $v_1,v_2,\dots,v_n$. If $s_1,s_2,\dots,s_n$ do not use up the budget $b$, i.e., $\sum_{i=p+1}^{n} \frac{d_{S \cup V_{i-1}}(i)}{d_S(i)}\cdot c(i) < b$, the claim is straightforward since $\sum_{i=1}^{n} v_i = \sum_{i=1}^{n} f_{S \cup V_{i-1}}(i) = f_S(V_{n})$ must be the maximum possible due to condition (iv) of Definition \ref{def:partitionnew} (by setting $i=n$). Now assume that $\sum_{i=1}^{n}s_i = b$. Let $t$ be the lowest index satisfying $\sum_{i=1}^{t}s_i = b$. By the procedure for computing $\Lambda^2(S)$, 
\begin{itemize}
\item for each $p+1 \le i < t$, $s_i = \frac{d_{S \cup V_{i-1}}(i)}{d_S(i)}\cdot c(i)$ and $\sum_{j=1}^{i}v_j = f_S(V_{i})$; 
\item $s_t = \left(b-\sum_{j=p+1}^{t-1}\frac{d_{S \cup V_{j-1}}(j)}{d_S(j)}\cdot c(j)\right)\cdot d_S(t)$ and $\sum_{j=1}^{t}v_j = \Lambda^2(S)$; 
\item for each $i > t$, $s_i = 0$ and  $\sum_{j=1}^{i}v_j = \Lambda^2(S)$.
\end{itemize}

We prove the lemma by contradiction. Let $\tilde{v}$ be the maximum value that has a valid partition. Assume that $\tilde{v} > \Lambda^2(S)$. 
There can be multiple valid partitions of $\tilde{v}$. Among all valid partitions of $\tilde{v}$, we pick one with the lowest total cost. 
If there are multiple valid partitions with the same lowest total cost, we pick one whose cost sequence is the last in the lexicographic order (for any two cost sequences, we can check the lowest index where the costs differ; the sequence with a smaller cost at this index is placed before the other sequence according to the lexicographic order; the lexicographic order is a total order).
Let $\tilde{s}_1,\tilde{s}_2,\dots,\tilde{s}_n$ be the cost sequence and $\tilde{v}_1,\tilde{v}_2,\dots,\tilde{v}_n$ be the value sequence in this valid partition, where $\sum_{i=1}^{n}\tilde{v}_i = \tilde{v}$.

Now we check the difference between the two sequences. We iterate through the sequence $\langle {p+1},{p+2},\dots,n\rangle$ to find the first element $q$ where $s_q \neq \tilde{s}_q$. We can assume $q \le t$, since otherwise $\sum_{i=1}^{t}\tilde{s}_i = \sum_{i=1}^{t}s_i = b$ and hence $\tilde{s}_q = 0 = s_q$ must hold.

If $\tilde{s}_q > s_q$, we must have $q < t$, because $\tilde{s}_t \leq b - \sum_{i=1}^{t-1}\tilde{s}_i = \sum_{i=1}^{t}s_i - \sum_{i=1}^{t-1}\tilde{s}_i = s_t$. We can construct a new cost sequence $s_1,\dots,s_{q-1},s_q,\tilde{s}_{q+1},\dots,\tilde{s}_n$ and a new value sequence $v_1,\dots,v_{q-1},$ $\sum_{i=1}^{q}\tilde{v}_i-\sum_{i=1}^{q-1}v_i,\tilde{v}_{q+1},\dots,\tilde{v}_n$. We verify that they satisfy all the conditions of Definition \ref{def:partitionnew}.
Condition (i) holds because the total value remains unchanged.
Condition (ii) holds because the total cost $\sum_{i=1}^{q}s_i + \sum_{i=q+1}^{n}\tilde{s}_i < \sum_{i=1}^{q}\tilde{s}_i + \sum_{i=q+1}^{n}\tilde{s}_i \le b$. 
Condition (iii) obviously holds for any $i \neq q$ as the cost-value pair of index $i$ is inherited from an existing valid partition. For $i=q$, we have $d_S(q)\cdot s_q \ge v_q = \sum_{i=1}^{q}v_i-\sum_{i=1}^{q-1}v_i = f_S(V_q)-\sum_{i=1}^{q-1}v_i \ge \sum_{i=1}^{q}\tilde{v}_i-\sum_{i=1}^{q-1}v_i$.
Condition (iv) holds for any $i$ since the total value up to index $i$ is the same as an existing valid partition.
Thus, the new cost and value sequences form a valid partition. The total value of the new partition is the same as $\sum_{i=1}^{n}\tilde{v}_i$, but the total cost of the new partition is lower than $\sum_{i=1}^{n}\tilde{s}_i$, contradicting that $\tilde{s}_1,\tilde{s}_2,\dots,\tilde{s}_n$ have the lowest total cost among all valid partitions.

If $\tilde{s}_q < s_q$, we can also construct a new valid partition contradicting the selection of $\tilde{s}_1,\tilde{s}_2,\dots,\tilde{s}_n$. 
Notice that $\sum_{i=1}^{q-1}v_i = \sum_{i=1}^{q-1}d_S(i) \cdot s_i = \sum_{i=1}^{q-1}d_S(i) \cdot \tilde{s}_i \ge \sum_{i=1}^{q-1}\tilde{v}_i$, where the first equality follows from the definition of $\Lambda^2(S)$ and the inequality follows from condition (iii) of Definition \ref{def:partitionnew}. Hence,
\begin{align*}
    f_S(V_{q}) - \sum_{i=1}^{q-1}\tilde{v}_i &\ge f_S(V_{q}) - \sum_{i=1}^{q-1}v_i \\
    &= f_S(V_{q}) - f_S(V_{q-1}) \\
    &= d_{S \cup V_{q-1}}(q)\cdot c(q) \\
    &= d_S(q)\cdot s_q \\
    &> d_S(q)\cdot \tilde{s}_q.
\end{align*}
Thus, condition (iv) has slack for $i=q$ in the assumed valid partition of $\tilde{v}$. 
Together with $\tilde{s}_q < s_q$, it implies that $\tilde{s}_q$ can be increased without violating both conditions (iii) and (iv) for $i=q$.
It can also be inferred that there is positive cost spent on at least one element in $q+1, q+2, \dots, n$, because otherwise $\tilde{v} = \sum_{i=1}^{n}\tilde{v}_i = \sum_{i=1}^{q}\tilde{v}_i \le \sum_{i=1}^{q}d_s(i)\cdot \tilde{s}_i \le \sum_{i=1}^{q}d_s(i)\cdot s_i \le \Lambda^2(S)$, contradicting that $\tilde{v} > \Lambda^2(S)$. 
Let $m$ be the lowest index larger than $q$ satisfying $\tilde{s}_m > 0$.
If we move a small amount of cost $\epsilon > 0$ from $\tilde{s}_m$ to $\tilde{s}_q$, $\tilde{v}_q$ can be increased by $d_S(q)\cdot \epsilon$. On the other side, $\tilde{v}_m$ can be decreased by $d_S(q)\cdot \epsilon$ (which does not violate condition (iii) for $i=m$ since $d_S(q)\cdot \epsilon \ge d_S(m)\cdot \epsilon$). 
Meanwhile, condition (iv) is not violated for any $q < i < m$ because $\tilde{v}_i = 0$ for all such $i$ values.
In this way, we have constructed a new valid partition of $\tilde{v}$ with higher cost spent on element $q$, contradicting that $\tilde{s}_1,\tilde{s}_2,\dots,\tilde{s}_n$ are the last cost sequence in the lexicographic order.  

We have identified a contradiction in every case and completed the proof.
\end{proof}

\section{Equivalence between the optimum of Equation (\ref{lnp:lbd 2}) and $\Lambda^2(S)$}
\label{appendix_lp_lambda2}
\begin{lemma}
\label{lem:lnp lbd 2}
    For any value $v$ that has a valid partition $v_1,v_2,\dots,v_n$, there exists an indicator vector $\mathbf{x} = \langle x_1, x_2, \dots, x_n \rangle \in [0,1]^n$ that satisfies all the constraints of \Cref{lnp:lbd 2} and $\sum_{i \in V \setminus S} f_S(i) \cdot x_i = v$.
\end{lemma}
\begin{proof}
    We set $x_i = \frac{v_i}{f_S(i)}$ for each element $i \in V \setminus S$, and $x_i = 0$ for each element $i \in S$. 
    It is easy to observe that $\sum_{i \in V \setminus S} f_S(i) \cdot 
    x_i = v$ (recall that $v_i=0$ for each $i \in S$ in a valid partition). Now we proceed to check if the $\mathbf{x}$ constructed meets all the constraints of \Cref{lnp:lbd 2}.

    For the first constraint, for each element $i \in V$,
\begin{equation*}
\sum_{j=1}^{i} f_S(j) \cdot x_j = \sum_{j=1}^{i}  v_i \le f_S(V_{i}), 
\end{equation*}
where the inequality is due to condition (iv) of Definition \ref{def:partitionnew}.

    Let the 
    cost sequence of the valid partition be $s_1,s_2,\dots,s_n$.
    For the second constraint, 
\begin{equation*}
\sum_{i \in V} c(i) \cdot x_i = \sum_{i=p+1}^{n} \frac{v_i}{d_S(i)} 
\le \sum_{i=p+1}^{n} s_i 
\le b.  
\end{equation*}
where the two inequalities are due to conditions (ii) and (iii) of Definition \ref{def:partitionnew}. 

    In addition, for each element $i \in V \setminus S$,
\begin{equation*}
    0 \le x_i = \frac{v_i}{f_S(i)} \le \frac{s_i \cdot d_S(i)}{f_S(i)} = \frac{s_i}{c(i)} \le 1,
\end{equation*}
where the second inequality is due to condition (iii) of Definition \ref{def:partitionnew} and the third inequality is due to the constraint $s_i \le c(i)$ in Definition \ref{def:partitionnew}.
\end{proof}
\begin{lemma}
\label{lem:lnp lbd 2 part 2}
    For each indicator vector $\mathbf{x}  = \langle x_1, x_2, \dots, x_n \rangle \in [0,1]^n$ that satisfies all the constraints of \Cref{lnp:lbd 2}, $\sum_{i \in V \setminus S} f_S(i) \cdot x_i$ has a valid partition.
\end{lemma}
\begin{proof}
    Let $s_i = c(i) \cdot x_i$ and $v_i = f_S(i) \cdot x_i$. It is easy to observe that $\sum_{i=1}^{n}v_i=\sum_{i=1}^{n}f_S(i) \cdot x_i = \sum_{i \in V \setminus S} f_S(i) \cdot x_i$.
    Now we proceed to check if $v_i$ and $s_i$ satisfy all the other conditions of a valid partition in  Definition \ref{def:partitionnew}.

    For condition (ii),
\begin{eqnarray*}
    \sum_{i=1}^{n} s_i = \sum_{i=1}^{n} c(i) \cdot x_i \le b,
\end{eqnarray*}
where the inequality is due to the second constraint of \Cref{lnp:lbd 2}.

    For condition (iii), for each element $i \in V$, \begin{eqnarray*}
    d_S(i) \cdot s_i = f_S(i) \cdot x_i = v_i.
\end{eqnarray*}

    For condition (iv), for each element $i \in V$,
\begin{eqnarray*}
    \sum_{j=1}^{i} v_i = \sum_{j=1}^{i} f_S(j) \cdot x_j \le f_S(V_i),
\end{eqnarray*}
where the inequality is due to the first constraint of \Cref{lnp:lbd 2}.

    Thus, all the conditions of Definition \ref{def:partitionnew} are satisfied. So, $\sum_{i \in V \setminus S} f_S(i) \cdot x_i$ has a valid partition.
\end{proof}
    With Lemma \ref{lem:lnp lbd 2} and Lemma \ref{lem:lnp lbd 2 part 2}, we can conclude that the optimum of \Cref{lnp:lbd 2} is equal to the maximum value that has a valid partition, i.e., $\Lambda^2(S)$.

\section{$\Lambda^3$ in linear program presentation}
\label{sec:lbd 3}
    To write the calculation of $\Lambda^3$ as a linear program, we can further add the removing item to \Cref{lnp:lbd 2}.
\begin{align}
\label{lnp:lbd 3}
    \max & \quad \sum_{i \in V \setminus S} f_S(i) \cdot x_i - \sum_{i \in S} f_{V \setminus \{i\}}(i) \cdot (1-x_i)\\
    \text{subject to} & \quad \forall i \in V, \quad \sum_{j=1}^{i} f_S(j) \cdot x_j \le f_S(V_i), \nonumber \\
    & \quad \sum_{i \in V} c(i) \cdot x_i \le b. \nonumber 
\end{align}
The optimum of \Cref{lnp:lbd 3} is equal to $\Lambda^3(S)$.
This can be proved through a similar analysis to Appendix \ref{appendix_lp_lambda1}. 

\label{appendix_lbd_3_complexity}

    If the elements in $V \setminus S$ are arranged in non-ascending order of marginal density with respect to $S$, 
    it is easy to see that the optimum of the following linear program is equal to $\Lambda^2(S)$ defined in \Cref{eq:Lambda2}: 
\begin{align}
\label{lnp:lbd 2 v3}
    \max & \quad \sum_{i \in V \setminus S} f_S(i) \cdot x_i \\
    \text{subject to} & \quad \sum_{i \in V} c(i) \cdot x_i \le b, \nonumber \\
    & \quad \forall i \in S, \quad 0\le x_i \le 1, \nonumber \\
    & \quad \forall i \in V \setminus S, \quad 0 \le x_i \le \frac{f_{S \cup V_{i-1}}(i)}{f_S(i)}. \nonumber
\end{align}
    Thus, \Cref{lnp:lbd 2 v3} is a variation of calculating $\Lambda^2(S)$.

    If we add the removing item to this linear program, we can then get a variation of calculating $\Lambda^3(S)$ as follows:
\begin{align}
\label{lnp:lbd 3 v3}
    \max & \quad \sum_{i \in V \setminus S} f_S(i) \cdot x_i -\sum_{i \in S} f_{V \setminus \{i\}}(i) \cdot (1-x_i)\\
    \text{subject to} & \quad \sum_{i \in V} c(i) \cdot x_i \le b, \nonumber \\
    & \quad \forall i \in S, \quad 0\le x_i \le 1, \nonumber \\
    & \quad \forall i \in V \setminus S, \quad 0 \le x_i \le \frac{f_{S \cup V_{i-1}}(i)}{f_S(i)}. \nonumber
\end{align}
    Comparing \Cref{lnp:lbd 3 v3} with \Cref{lnp:lbd 1}, the only difference is the feasible range for $x_i$ where $i \in V \setminus S$. Similar to computing $\Lambda^{1}(S)$ based on \Cref{lnp:lbd 1}, we can compute $\Lambda^3(S)$ by picking elements greedily according to their ``weights'' in the objective function until the budget is fulfilled to get the optimum of \Cref{lnp:lbd 3 v3}, where for each $i \in S$, its weight is $f_{V \setminus \{i\}}(i)$; for each $i \in V\setminus S$, its weight is $f_S(i)$. Consequently, the time complexity to compute $\Lambda^3(S)$ is also $O(n \log n) = O(|V| \log |V|)$ due to the sorting of the elements, the same as $\Lambda^{0}(S)$, $\Lambda^{1}(S)$ and $\Lambda^{2}(S)$.

\section{Experimental results for additional cost settings}
\label{sec:add_cost_setting}

\Cref{fig:80 small 20 large} presents the approximation guarantees derived from different upper bounds for cost setting (b). We conduct experiments for all ground sets generated with different budgets from 20 to 50 (step size 1). So, for each data set, there are 620 problem instances (20 ground sets $\times$ 31 different budgets from 20 to 50). Following the recommendations in \Cref{sec:exp}, our evaluation includes all enumerated augmentations $\Lambda^{0+}$, $\Lambda^{1+}$, $\Lambda^{2+}$, $\Lambda^{3+}$, and two unified augmentations $\Lambda^{0*}$, $\Lambda^{1*}$. 

\Cref{fig:20 small 80 large} presents the approximation guarantees derived from different upper bounds for cost setting (c). We conduct experiments for all ground sets generated with different budgets from 6 to 20 (step size 1). So, for each data set, there are 300 problem instances (20 ground sets $\times$ 15 different budgets from 6 to 20). 

\begin{figure}[!t]

\centering
\includegraphics[width=0.7\textwidth]{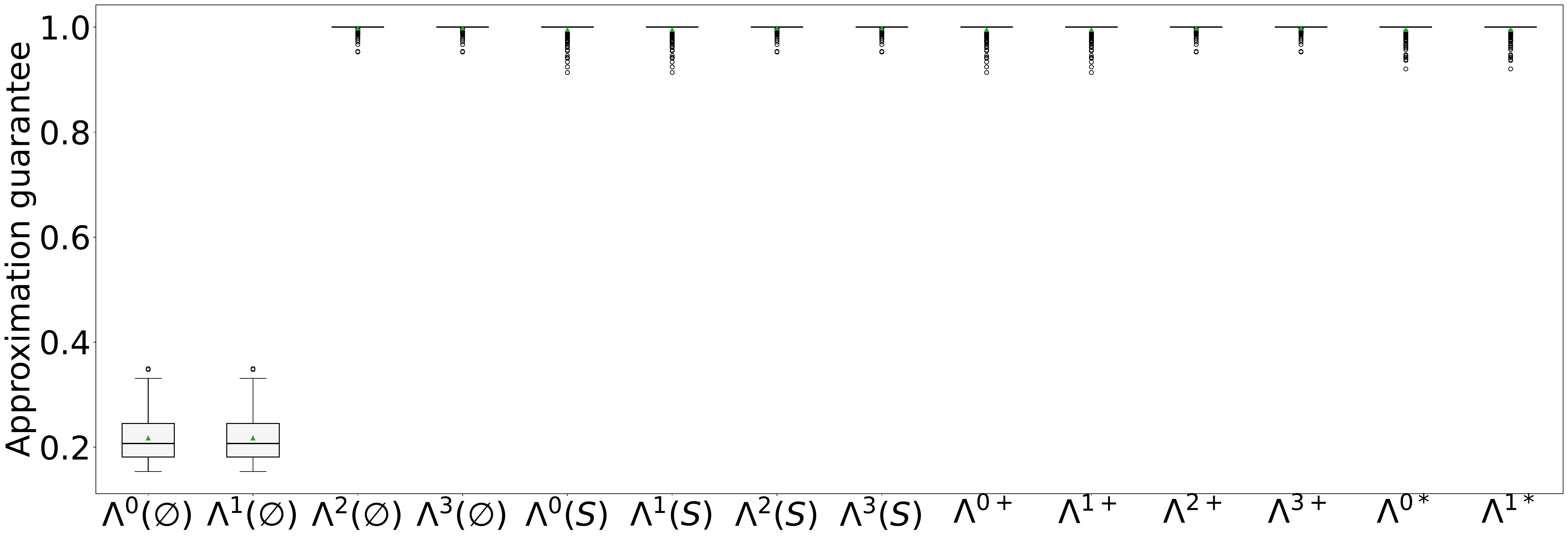} \\
(a) Adult Income 

\centering
\includegraphics[width=0.7\textwidth]{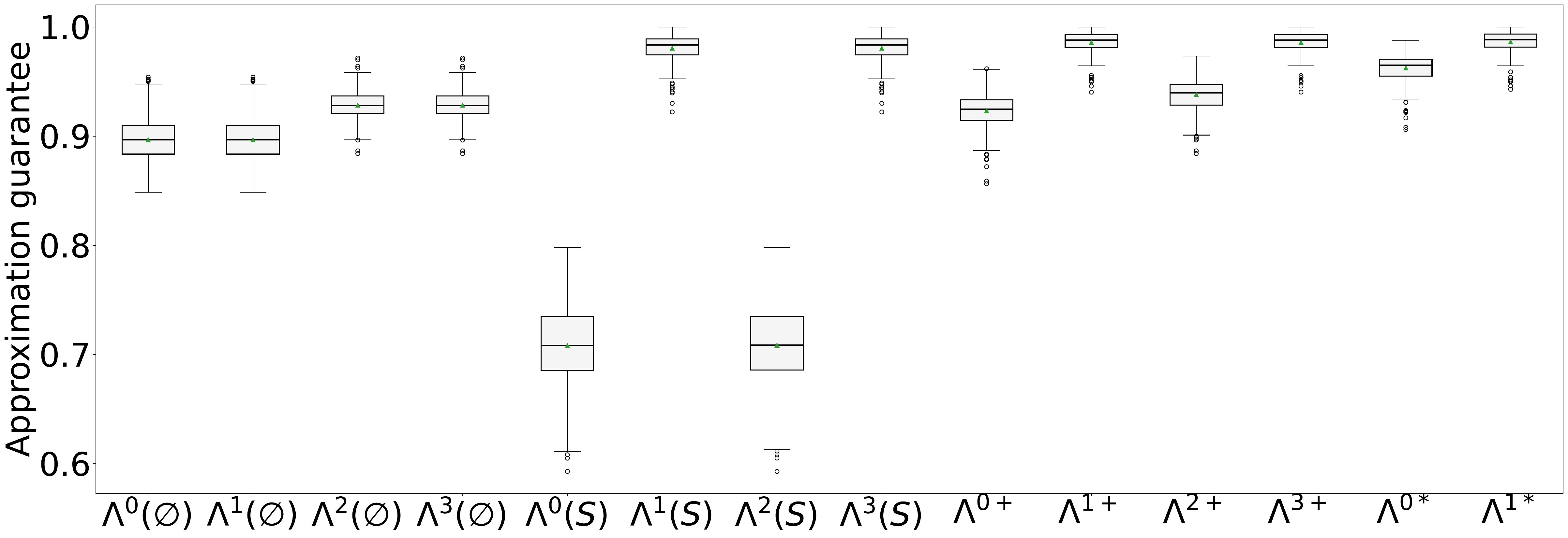} \\
(b) Caltech36 

\centering
\includegraphics[width=0.7\textwidth]{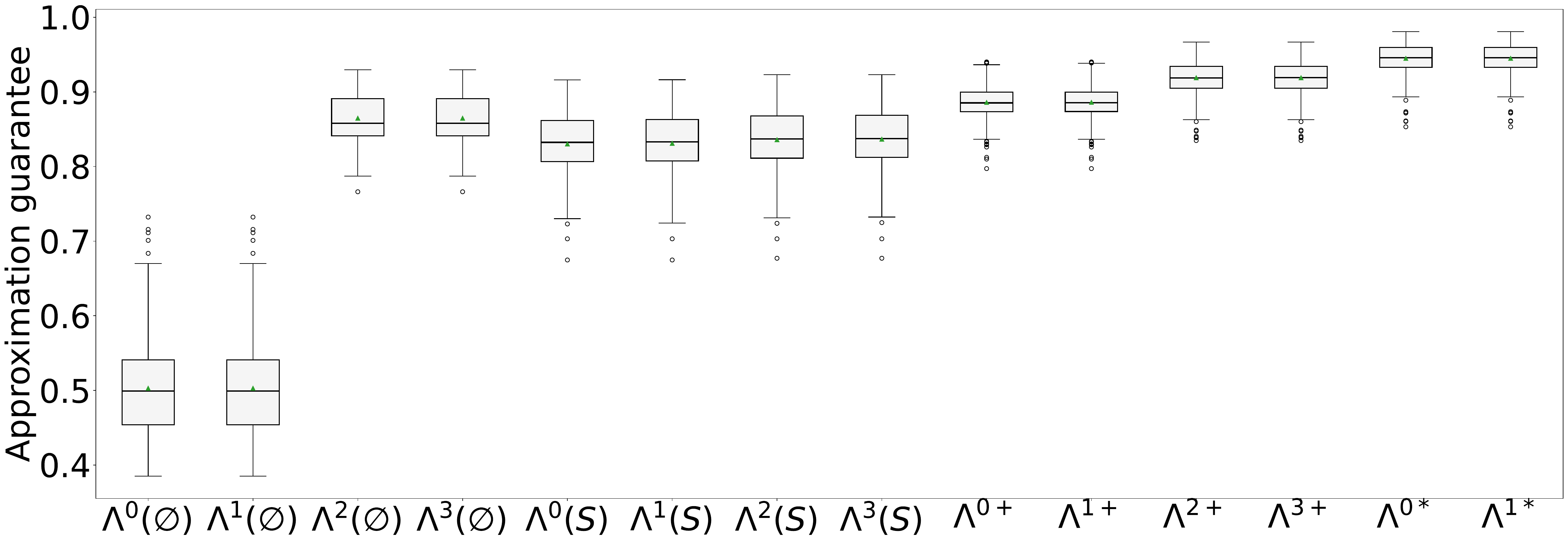} \\
(c) ego-facebook 

\centering
\includegraphics[width=0.7\textwidth]{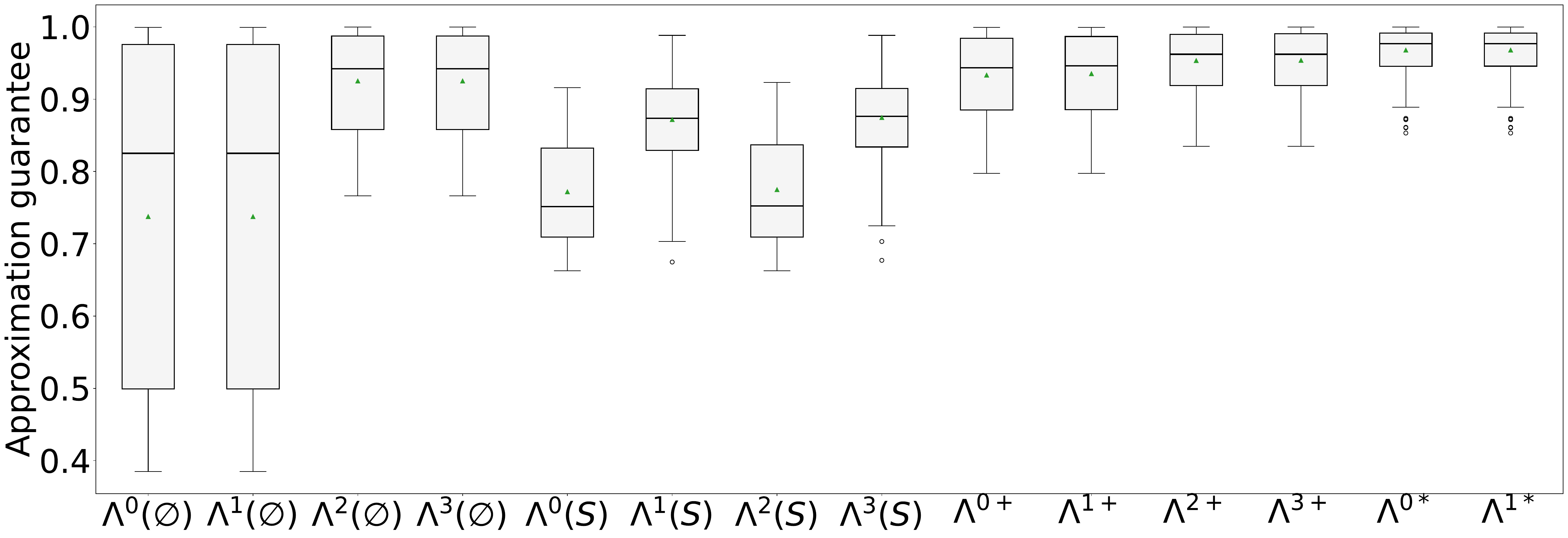} \\
(d) com-youtube 

\caption{Actual approximation guarantee plots for different upper bounds in cost setting (b)} 
\label{fig:80 small 20 large}
\end{figure}

\begin{figure}[!t]

\centering
\includegraphics[width=0.7\textwidth]{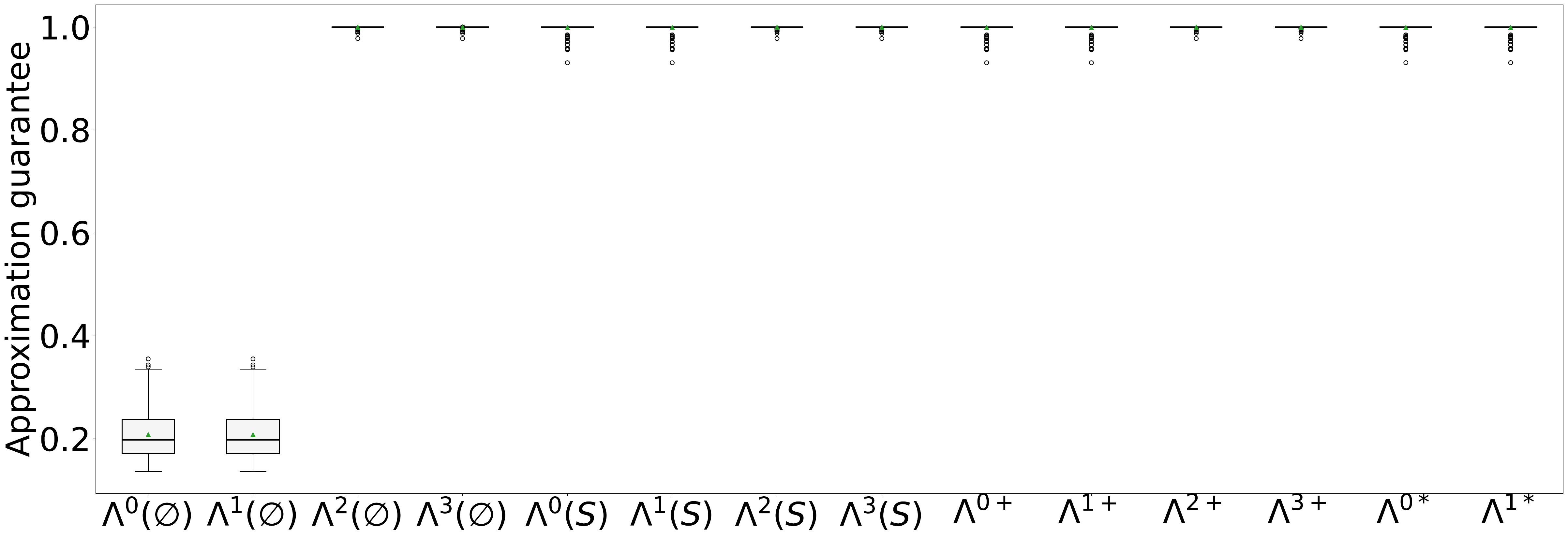} \\
(a) Adult Income 

\centering
\includegraphics[width=0.7\textwidth]{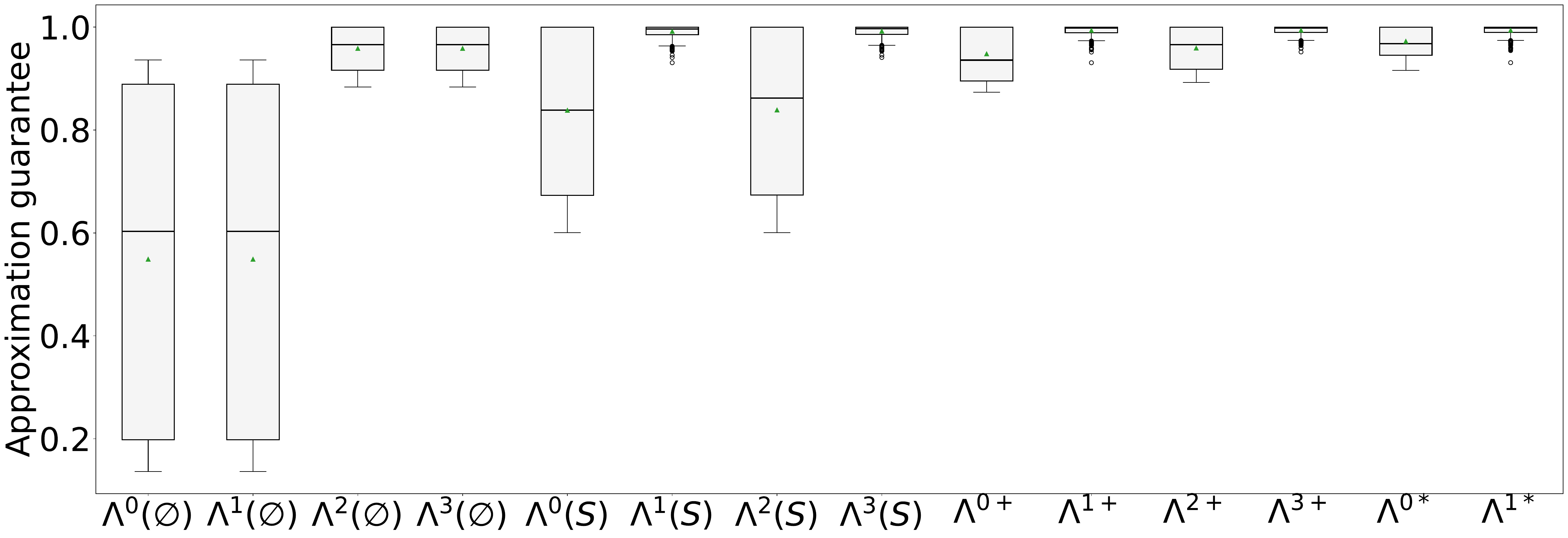} \\
(b) Caltech36 

\centering
\includegraphics[width=0.7\textwidth]{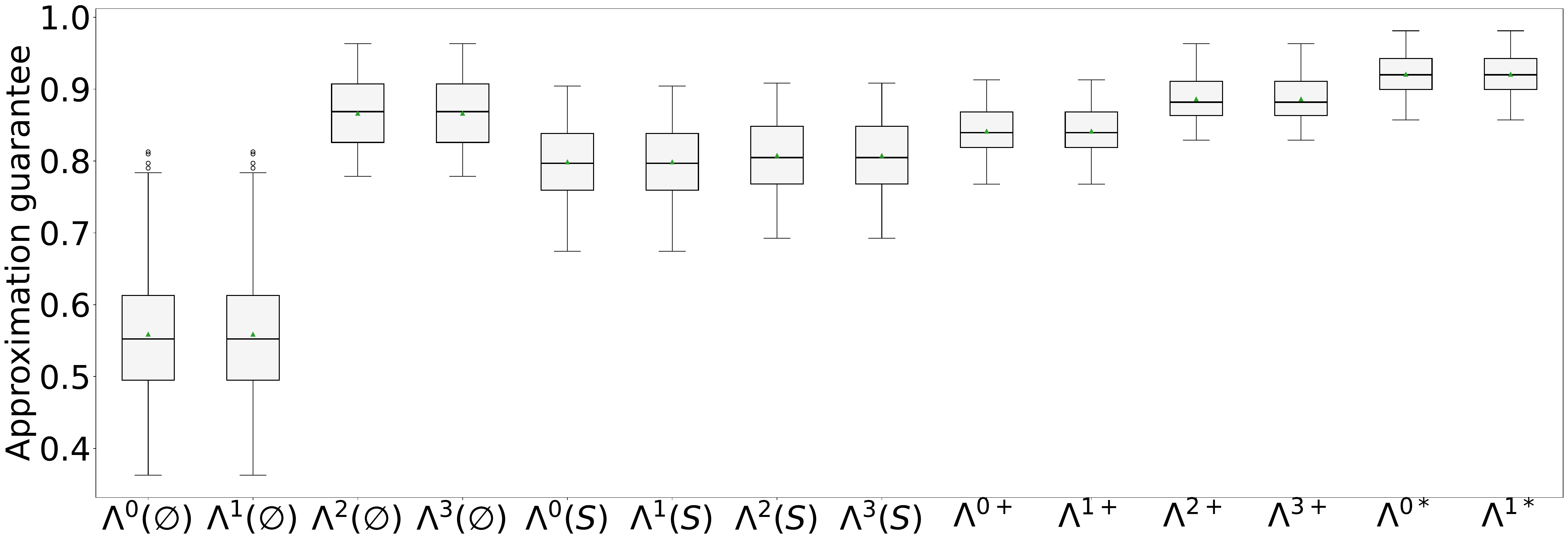} \\
(c) ego-facebook 

\centering
\includegraphics[width=0.7\textwidth]{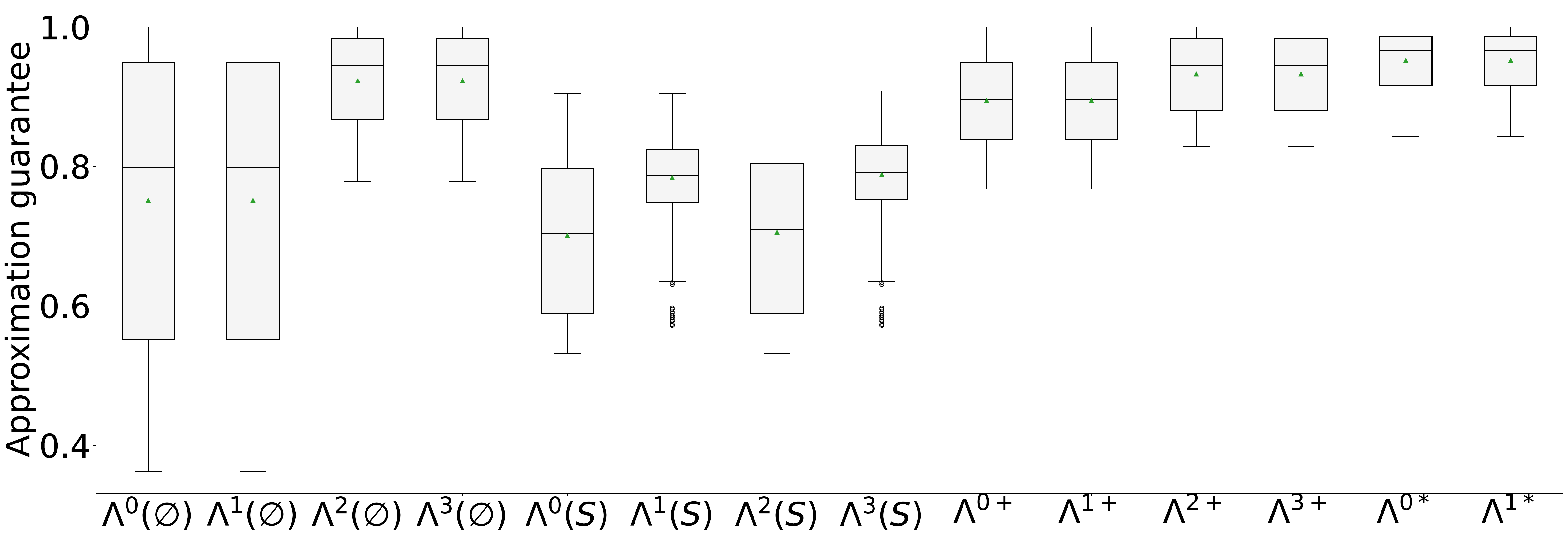} \\
(d) com-youtube 

\caption{Actual approximation guarantee plots for different upper bounds in cost setting (c)} 
\label{fig:20 small 80 large}
\end{figure}

\end{document}